\documentclass{article}
\usepackage[all]{xy}

\usepackage{amsmath}
\usepackage{amscd}
\usepackage{amsthm}
\usepackage{amssymb}
\usepackage{amsfonts}
\usepackage{enumerate}
\usepackage{graphicx, color}
\usepackage[font=normalsize]{caption}
\usepackage[mathcal]{euscript}
\usepackage{mathtools}
\usepackage{mathrsfs}
\usepackage{tikz}
\usetikzlibrary{arrows,decorations.markings}
\usepackage{subfig}
\usepackage{indentfirst} 
\usepackage{xcolor}
\usepackage{tikz-cd}

\numberwithin{equation}{section}

\theoremstyle{theorem}
\newtheorem{thm}{Theorem}[section]
\newtheorem{cor}{Corollary}[section]
\newtheorem{lem}{Lemma}[section]
\newtheorem{prop}{Proposition}[section]

\theoremstyle{definition}
\newtheorem{defn}{Definition}[section]
\newtheorem{rem}{Remark}[section]
\newtheorem{hypo}{Hypothesis}[section]
\newtheorem{notation}{\rm\bfseries{Notation}}[section]

\newtheorem{ques}{Question}[section]

\newtheorem{exm}{Example}[section]

\numberwithin{equation}{section}

\numberwithin{equation}{section}

\def\R{{\mathbb R}}

\def\SPS{{\hbox{\rm SPS}}}
\def\TPS{{\hbox{\rm TPS}}}
\def\KTPS{{\hbox{\rm KTPS}}}

\def\E{{\mathbb E}}
\def\Z{{\mathbb Z}}
\def\C{{\mathbb C}}
\def\R{{\mathbb R}}
\def\P{{\mathbb P}}

\def\N{{\mathbb N}}

\def\11{{\mathbb I}}

\def\del{{\partial}}

\def\C{\mathbb{C}}
\def\Z{\mathbb{Z}}

\def\Q{\mathbb{Q}}

\def\E{\ifmmode{\mathbb E}\else{$\mathbb E$}\fi} 
\def\N{\ifmmode{\mathbb N}\else{$\mathbb N$}\fi} 
\def\R{\ifmmode{\mathbb R}\else{$\mathbb R$}\fi} 
\def\Q{\ifmmode{\mathbb Q}\else{$\mathbb Q$}\fi} 
\def\C{\ifmmode{\mathbb C}\else{$\mathbb C$}\fi} 
\def\H{\ifmmode{\mathbb H}\else{$\mathbb H$}\fi} 
\def\Z{\ifmmode{\mathbb Z}\else{$\mathbb Z$}\fi} 
\def\P{\ifmmode{\mathbb P}\else{$\mathbb P$}\fi} 
\def\SS{\ifmmode{\mathbb S}\else{$\mathbb S$}\fi} 
\def\DD{\ifmmode{\mathbb D}\else{$\mathbb D$}\fi} 

\def\R{{\mathbb R}}

\def\E{{\mathbb E}}
\def\Z{{\mathbb Z}}
\def\C{{\mathbb C}}
\def\R{{\mathbb R}}

\def\N{{\mathbb N}}

\def\CA{{\mathcal A}}
\def\CB{{\mathcal B}}

\def\CD{{\mathcal D}}
\def\CE{{\mathcal E}}
\def\CF{{\mathcal F}}
\def\CG{{\mathcal G}}

\def\CJ{{\mathcal J}}

\def\CL{{\mathcal L}}
\def\CM{{\mathcal M}}

\def\CO{{\mathcal O}}
\def\CP{{\mathcal P}}

\def\CP{{\mathcal P}}
\def\CS{{\mathcal S}}
\def\CT{{\mathcal T}}

\def\opname#1{\mathop{\kern0pt{\rm #1}}\nolimits}

\def\dim{\opname{dim}}

\def\group#1{\opname{#1}}

\def\U{\group{U}}

\def\Sing{\operatorname{Sing}}

\def\Diff{\operatorname{Diff}}
\def\Image{\operatorname{Image}}
\def\ev{\operatorname{ev}}

\def\Symp{\operatorname{Symp}}

\def\ben{\begin{enumerate}}
\def\een{\end{enumerate}}
\def\be{\begin{equation}}
\def\ee{\end{equation}}
\def\bea{\begin{eqnarray}}
\def\eea{\end{eqnarray}}
\def\beastar{\begin{eqnarray*}}
\def\eeastar{\end{eqnarray*}}
\def\bc{\begin{center}}
\def\ec{\end{center}}

\def\Sing{\operatorname{Sing}}

\def\dim{\opname{dim}}

\def\group#1{\opname{#1}}

\def\U{\group{U}}

\def\Diff{\operatorname{Diff}}

\title{Nonequilibrium thermodynamics as a symplecto-contact reduction and
relative information entropy}

\author{ First Author \\ Postal Address \\ e-mail:  \\[2ex]
         Second Author\thanks{ Supported by ... }
                      \\ Postal Address \\ e-mail:  }

\author{Jin-wook Lim\\ 
Department of Mathematics POSTECH \& \\
Center for Geometry and Physics, Institute for Basic Science (IBS), \\
79 Jigok-ro, 127beong-gil, Nam-gu, Pohang-si, \\
Gyeongsangbuk-do, Korea 37673 \\
jinwook.lim@gmail.com: \\[2ex]
Yong-Geun Oh \\
Center for Geometry and Physics, Institute for Basic Science (IBS), \\
79 Jigok-ro, 127beong-gil, Nam-gu, Pohang-si, \\
 Gyeongsangbuk-do, Korea 37673 \& \\
Department of Mathematics, POSTECH, \\
yongoh1@postech.ac.kr}

\begin{document}

\maketitle

\begin{abstract}
Both statistical phase space (SPS), which is
$\Gamma = T^*\mathbb R^{3N}$ of $N$-body particle system,
and kinetic theory phase space (KTPS), which is the cotangent bundle $T^*\mathcal P(\Gamma)$ of the probability space $\mathcal P(\Gamma)$
thereon, carry canonical symplectic structures. Starting from this first principle, we provide a canonical derivation of
thermodynamic phase space (TPS) of nonequilibrium thermodynamics
 as a contact manifold in two steps. First, regarding the collective observation of observables in SPS as a moment map defined on KTPS, we apply the Marsden-Weinstein reduction
 and obtain a mesoscopic phase space in between KTPS and TPS as a (infinite dimensional) symplectic fibration.
Then we show that the reduced relative information entropy 
defines a generating function that provides a covariant construction of a thermodynamic equilibrium
as a Legendrian submanifold. This Legendrian submanifold is not necessarily graph-like.
We interpret the Maxwell construction of \emph{equal-area law} as the procedure of finding a continuous, 
not necessarily differentiable,
thermodynamic potential and explain the associated phase transition by identifying the procedure
with that of finding a graph selector in symplecto-contact geometry and in the Aubry-Mather theory of dynamical
system.
\end{abstract}


\section{Introduction}

Starting from Caratheodory \cite{caratheodory} and Hermann \cite{hermann:geophy},
contact geometry is utilized in the geometric formulation of the thermodynamics, especially of equilibrium
thermodynamical process. It has been observed that the state of
thermodynamic equilibrium can be interpreted as a Legendrian submanifold
in the \emph{thermodynamic phase space} (TPS) equipped with a contact form.
The expression of the contact form is not unique in that
it depends on the choice of coordinates.  In this regard, more fundamental geometric
structure is not the contact form but the contact structure. We call
this contact structure the \emph{thermodynamic contact structure}.

In fact, TPS exists only abstractly which has been described
through coordinates in the literature. One of the main goals of the present paper is
to derive TPS  as a contact manifold by a
canonical and covariant reduction process starting from
the \emph{kinetic theory phase space} (KTPS), which is the cotangent bundle
$$
T^*\CP(\Gamma)
$$
of the space $\CP(\Gamma)$ of probability distributions
of the \emph{statistical phase space} (SPS) $\Gamma$ which is the cotangent
bundle $T^*\R^{3N}$ of $\R^{3N}$ (or its open subset). The latter is nothing but the
$N$ body particle system's phase space.

Such kind of an effort has been made by Grmela \cite{grmela:mesoscopic}.
This is manifest in the following phrase he made in \cite[p.1666]{grmela:mesoscopic}:
\medskip

\emph{``Again, the main difficulty on this route is in the reduction process.
The question that arises on both routes is the following: how shall we make the
reductions in order that the essential features of the mechanics put into the
original microscopic formulation is kept intact and the compatibility with thermodynamics
of both microscopic and mesoscopic formulations is guaranteed.''}

\medskip

One of the main results of the present paper is to derive this reduction process described
in this quoted phrase as \emph{a symplecto-contact reduction in the symplecto-contact Hamiltonian geometry}
focusing on the conceptual framework without delving into the functional analytical issues
arising from dealing with the infinite dimensional spaces such as $\CP(\Gamma)$, $T^*\CP(\Gamma)$ and others. 
For example, we will ignore the fact that the
canonical symplectic form of $T^*\CP(\Gamma)$ is only a \emph{weak} symplectic form
which is \emph{weakly nondegenerate}. (See \cite{abraham-marsden} for its definition.)

The canonical reduction process in symplectic geometry is introduced
by Marsden and Weinstein  \cite{marsden-weinstein} when there is a Hamiltonian
$G$ action by a Lie group $G$ on a symplectic manifold $(M,\omega)$.
 The notion of moment map is crucial in this reduction
process that is also introduced therein as a $G$-equivariant map
$$
J: M \to {\mathfrak g}^*
$$
where ${\mathfrak g}^*$ is the dual of the Lie algebra $\mathfrak g = \text{\rm Lie}(G)$ equipped with the coadjoint action.

We apply the reduction process to the \emph{infinite
dimensional} symplectic manifold $T^*\CP(\Gamma)$
under the canonical symplectic action induced by certain
subgroup of symplectomorphism group $\Symp(T^* \CP(\Gamma))$ that canonically
arise from the given \emph{system of observables} on the particle phase space, SPS.

\subsection{Relative information
entropy and thermodynamic entropy}

We remark that the particle phase space $\Gamma := T^*\R^K$, which we call $\SPS$,
 of the gas (i.e., many particle systems)
carries the canonical Lebesgue measure $\nu_0:= d\Gamma$ as the reference measure.
Here $K = 3N$ where $N$ is the number of particles.
A positive density $D$  on the space
$\Gamma$ can be expressed as
$$
D = f\, d\Gamma, \quad \text{for } \, \, f \geq 0 \,\, \,  \& \, \int_\Gamma f \, d\Gamma > 0
$$
where the function $f$ is the Radon-Nikodym derivative
$$
f = \frac{\del D}{\del{\nu_0}}.
$$
We denote by $\CD^+(\Gamma)$ the set of positive densities.

Our reduction process can be viewed as a statistical mechanics derivation
of thermodynamics starting from the kinetic approach to the study of many
particle systems such as gas which is going back to Bernoulli and Clausius.
It was Maxwell \cite{maxwell} who first laid down the true base thereof. (See \cite[Section 1.1]{villani:kinetic}
for a nice historical introduction of the kinetic approach.)
This kinetic approach makes appearance of cotangent bundle
 of $\CP(\Gamma)$  in the story natural.

 Another natural source of motivation regarding $\CP(\Gamma)$ as the  starting point arises
 from the relationship between Shannon's information entropy and the thermodynamic
 entropy.  The Shannon entropy is defined for a continuous observable (or a random variable) with
 a probability distribution $\rho$ on $\Gamma$, i.e., as a function on $\CP(\Gamma)$, and
 the thermodynamic entropy can be canonically derived through some `reduction' process.
 (See \cite{Jaynes:InfoStat}, \cite{MNSS} or Section \ref{sec:review} for such a derivation.)

Information entropy is defined for probability measures, i.e., positive densities of mass 1.
(See Appendix A for a brief recollection of basic definitions from Information theory.)

\begin{defn}[Information density] Let $\rho \in \CP(\Gamma)$ be a probability density.
We define the \emph{relative information density}, denoted by $I = I(\rho)$,
of $\rho$ to be
\begin{equation}\label{eq:gibbs-entropy}
I(\rho) = - \log \left(\frac{\del \rho}{\del \nu_0}\right)
\end{equation}
for the Radon-Nikodym derivative $\frac{\del \rho}{\del \nu_0}$. The \emph{(relative) information entropy} is  the function
$\CS: \CP(\Gamma) \to \R$ by
\begin{equation}\label{eq:CS-intro}
\CS(\rho): = \int_\Gamma I(\rho) \rho  \left( = - D_{KL}(\rho\| \nu_0)\right).
\end{equation}
More generally we define the \emph{information measure} $I_\rho(A)$ of a 
measurable subset $A \subset \Gamma$ to the integral $I_\rho(A) = \int_A I(\rho)\, \rho$.
\end{defn}

\begin{rem}
We refer to \cite{kullback-leibler}, \cite{MACKAY}
for the general properties of \emph{Kullback-Leibler divergence} 
$D_{KL}(\rho\| \nu_0)$. We will not use this term beyond the above definition except
when we want to emphasize `being not Shannon's original information entropy but a relative information
entropy'. We will show that this relative information entropy
 \emph{$\CS$ is preserved under the action by $\nu_0$-measure preserving
diffeomorphisms.} This invariance will be of fundamental importance in our derivation later of
nonequilibrium thermodynamics as a symplecto-contact reduction of statistical mechanics.
\end{rem}

\subsection{Action of diffeomorphism group}

Let ${\rm Diff}(\Gamma)$ be the diffeomorphism group of SPS
$\Gamma= T^*\R^{3N}$. We equip $\Gamma$ with the canonical
symplectic form $\omega_\Gamma = \Omega_0$. Then
the Lebesque measure $\nu_0$ coincides with
the Liouville measure of the symplectic manifold
$(\Gamma,\Omega_0)$ in that $\nu_0$ is the integration measure of
the Liouville volume form
$$
\frac{1}{(3N)!} \omega_\Gamma^{3N}.
$$
The (extended) entropy function canonically arise from the probability
density, \emph{independently of the choice of other observables}. It
 is invariant under the action of the subgroup
$$
\widehat{\Diff}(\Gamma) \subset \Symp(T^* \CP(\Gamma))
$$
on KTPS, which consists of the symplectomorphisms
on $T^*\CP(\Gamma)$ naturally lifted from
the action of $\Diff(\Gamma)$ on the base $\CP(\Gamma)$. It is given
 by the
pushforward operation of the \emph{microscopic group action}
thereof on SPS.
We will also consider the microscopic action of the subgroup $\Symp(\Gamma,\Omega_0) \subset \Diff(\Gamma)$
 consisting of symplectic diffeomorphisms of the symplectic manifold
$(\Gamma,\Omega_0)$, and has its lifting of
$$
\widehat{\Symp}(\Gamma,\Omega_0) \subset \widehat{\Diff}(\Gamma)
$$
acted upon $T^*\CP(\Gamma)$.
We will derive some consequences
arising from the canonical action of $\widehat{\Symp}(\Gamma,\Omega_0) $,  
especially from the \emph{measure-preserving} property
of the action of $\Symp(\Gamma,\Omega_0)$  on $\Gamma$.

 In the statistical mechanics derivation of thermodynamics,
points in thermodynamic phase space represent
macroscopic \emph{observations} or \emph{measurements} of the relevant microscopic system, which we will
denote by $\CF$, through the averaging process.
In the kinetic theory a state is represented by a distribution density $D$
(say, in Boltzmann's study of the kinetic theory). See Definition
\ref{defn:observation} below for the precise definition of what we
call the \emph{observation}.

Therefore
we may regard  the macroscopic observation or the measurement
in the experiment of the given statistical system as a reduction process 
of the mechanics via N\"other's principle using the infinite dimensional symmetry of measure-preserving diffeomorphisms. 
Our first motto is

\begin{center}
\bf ``Observation is a moment map.''
\end{center}

The usual TPS is then the outcome of some kind of reduction of KTPS, where
the relative information entropy function \emph{universally arises} as a reduced function
on any thermodynamic phase space, i.e., in any thermodynamic
model: \emph{It does not carry its conjugate partner}
in the nonequilibrium thermodynamic phase space $\TPS$,
unlike other observables. This will be the source of \emph{odd dimensionality}
of TPS.
\emph{
Furthermore the temperature $T$, which is usually regarded as
the conjugate of entropy, does not appear in $\SPS$
until the state reaches a thermodynamic equilibrium. }

A thermodynamic equilibrium is  defined to be the state that realizes the
maximum of the entropy in a given system. Traditionally in the thermodynamics literature,
the equilibrium is identified by first solving a constrained extremal problem
arising from the given observation data and the obvious constraint
\begin{equation}\label{eq:int-rho=1}
\int \rho = 1
\end{equation}
inside the space $\CD^+(D)$ of positive densities, and then carrying out the
macroscopic thermodynamic analysis.
To solve the first problem, the method of \emph{Lagrange multipliers}
has been utilized in the thermodynamic literature (e.g. \cite{Mrug:TPS})
in terms of \emph{preferred variables}.
We will geometrize and globalize this constrained extremal problem as a part of our
reduction process and solve the extremal problem \emph{globally and covariantly}
without using coordinates. In this procedure, we perform two constructions in order:
\begin{enumerate}[(1)]
\item {\bf Step 1:} We apply Marsden-Weinstein reduction for each collective observation of the given statistical system, denoted by $\CF$,
and obtain an intermediate reduced space $\CM^\CF$
as a symplectic fibration over the space of observation data.
\item {\bf Step 2:} We employ  the \emph{method of generating functions} in symplecto-contact
geometry by taking the reduced function, $\CS^{\text{\rm red}}_\CF$, as the global canonical generating
function of  the thermodynamic equilibrium state.
(See Appendix \ref{sec:generating} for the definition.)
\end{enumerate}
More detailed explanation of these two operations are now in order.

\subsection{Covariant construction of thermodynamic equilibrium}
\label{subsec:thermodynamic-reduction}

We start with what we mean by a (microscopic) statistical system $\CF$.
We call a collection of observables, i.e., of functions
$$
\CF =  \{F_1, \ldots, F_n\}; \quad F_i: \Gamma = T^*\R^K \to \R
$$
a \emph{statistical system} on SPS.
We consider the symmetry group consisting of
the symplectic diffeomorphisms of $T^*\CP(\Gamma)$ induced from the microscopic action of the
subgroup $\Diff(\Gamma, \nu_0) \subset \Diff(\Gamma)$  on the base $\Gamma$: The subgroup $\Diff(\Gamma, \nu_0)$
 consists of $\nu_0$-measure preserving diffeomorphisms, e.g.,
symplectic diffeomorphisms of $(\Gamma,\Omega_0)$.

\begin{rem}\label{rem:kapranov}
Kapranov \cite{kapranov:thermodynamics}  looked at the aforementioned
statistical system in the point of view of thermodynamics \emph{with several
Hamiltonians}. He also related them to the toric moment map for the
case of commuting Hamiltonians. Here we single out the Hamiltonian as the observable that
drives the dynamics and other ones as the collection that define the given statistical
system such as volume, mole number and other observables. However our first motto
``Observation is a moment map." is in the similar spirit as that of \cite{kapranov:thermodynamics}.
\end{rem}

We find it convenient for the further discussion to introduce the following intuitive terminology.
\begin{defn}[Observation (aka Measurement)]
For a given local observable $F$ on SPS,
we define its \emph{observation}, denote by $\CO_F$,
to be the function of taking the macroscopic average
$$
\CO_F: \CP(\Gamma) \to \R
$$
given by $\CO_F(\rho) = \int_\Gamma F \, \rho$.  We will equally
use the term \emph{measurement}  with the same
meaning as the term
\emph{observation}, depending on the given circumstances.
\end{defn}
 
Note that as a statistical observable, our relative information density
$I(\rho) = -\log \frac{\del \rho}{\del \nu_0}$,
is a universal observable in SPS (as a measure space
equipped with the Liouville measure) whose observation is the relative information entropy.
Because of the constraint $\int_\Gamma \rho = 1$, it will not
carry any thermodynamic conjugate partner  in our reduction
process. This is responsible for the odd dimensionality of TPS:
\emph{the $\CF$-reduced entropy is
 the generating function of a thermodynamic equilibrium, a Legendrian  submanifold in the one-jet bundle.}

The aforementioned action of $\Diff(\Gamma, \nu_0)$ preserves
both the relative information entropy \eqref{eq:CS-intro} and
the observations of observables of the given statistical system 
$\CF$ which enables us to apply the standard Marsden-Weinstein reduction:
Their collective observations denoted by
$$
\CO_\CF = \{\CO_{F_1},\cdots \CO_{F_n}\}: \CP(\Gamma) \to \R^n
$$
canonically lift to the functions
$$
\widetilde \CO_\CF = \{\widetilde{\CO}_{F_1}, \cdots, \widetilde{\CO}_{F_n}\}: T^*\CP(\Gamma) \to \R^n
$$
which are Poisson-commuting.
We then formally extend the symplectic reduction construction to the infinite dimensional
\emph{cotangent bundle} $T^*\CP(\Gamma)$ of the set $\CP(\Gamma)$
of probability densities on $\Gamma$, as usually done in the Hamiltonian formalism of the field theory.
See Remark \ref{rem:bracket} for the expression of the canonical
symplectic structure on $T^*\CP(\Gamma)$.  (We highlight the
fact that we consider a field theory on the \emph{affine space} of
$\CP(\Gamma)$, \emph{not} on a vector space on which a
usual Hamiltonian formalism of a field theory takes place.)

Regarding $\widetilde \CO_\CF$ as a moment map, we
apply the Marsden-Weinstein reduction to each (regular) value of the moment map $\widetilde{\CO}_\CF$.
This leads us to the construction of a (infinite dimensional) \emph{symplectic fiber bundle}
$$
\CM^\CF \to B_{\CF}^\circ \subset \R^n
$$
(over some open subset $B_{\CF}^\circ$ consisting of regular values of $\widetilde{\CO}_{\CF}$). 
We regard this collective reduced space as
the \emph{mesoscopic phase space} in between KTPS and TPS
on which the canonical reduction, denoted by $\CS_{\CF}^{\text{\rm red}}$,
of the function $\CS$ is defined.

In the literature of thermodynamics,  the (equilibrium) thermodynamics  
 is attempted to be derived  from the first principle of statistical 
 mechanics utilizing the information entropy in the point of view of Shannon's information theory.
 For example in \cite{Jaynes:InfoStat}, \cite{Mrug:TPS}, \cite{MNSS},
the authors describe a thermodynamics equilibrium in terms of the
information entropy and contact geometry. However the description is not satisfactory enough in that
in their derivation it is assumed  that the associated Legendrian submanifold is \emph{graph-like}
relative to the given variables, i.e., it is assumed that
the Legendrian submanifold $R_{\CF;\CS}$ admits a single-valued potential
function $f: B_{\CF} \to \R$ such that
$$
R_{\CF;\CS} = \{(\mu, df(\mu), f(\mu)) \mid \mu \in B_{\CF}\}.
$$
\begin{defn}[$C^0$-holonomic] Let $B$ be a smooth (possibly infinite
dimensional) manifold. We say a Legendrian submanifold
$R \subset J^1B$ is \emph{$C^0$-holonomic} if it is graph-like, i.e., if
it is  the image of the 1-jet map $j^1f: B \to J^1B$ of a differentiable function $f: B \to \R$.
\end{defn}

\begin{rem} Our usage of the term $C^0$-holonomic 
is the $C^0$-version of the standard notion of
\emph{holonomic section} $\varphi:V \to J^rX$ which is the $k$-jet
of some $C^k$-section $f: V \to X$ for the given $C^\infty$-fibration.
(See \cite[Section 1.1.1]{gromov:relation} for the definition.)
In our current situation, we consider the 1-jet bundle $J^1B \to B$
and regard $R \subset J^1B$ as a (possibly partially defined)
multi-valued section $B \to J^1B$.  This terminology $C^0$-holonomic being
introduced, we will just simply call \emph{holonomic}
instead throughout  the paper when  $R$ is graph-like.
\end{rem}

The following globalization of the description of the thermodynamic equilibrium complements 
the ones described in \cite{Jaynes:InfoStat}, \cite{Mrug:TPS}, \cite{MNSS} by allowing non-holonomic
equilibrium states in the description thereof. This leads us to our second motto:

\begin{center}
\bf ``Relative information entropy is \\
the generating function
of thermodynamic equilibria.''
\end{center}

Now we can turn the formal restatement briefly mentioned at the end of the previous subsection into the following precise mathematical statement.

\begin{thm} [Universality of information entropy]\label{thm:universality}
Let $\CF = \{F_1, \cdots, F_n\}$ be any observable system on SPS.
\begin{enumerate}
\item The \emph{lifted} collective observation
$$
\widetilde{\CO}_\CF = (\widetilde{\CO}_{F_1}, \cdots,  \widetilde{\CO}_{F_n})
$$
is the moment map whose reduced spaces define a symplectic
fiber bundle $\CM^\CF$ over an open subset of  $\R^n$, which
consists of the set of regular values of the collective observation
map $\CO_\CF$.

\item The \emph{lifted} relative entropy function $\widetilde \CS: T^*\CP(\Gamma) \to \R$ induces
a generating function $\CS_\CF^{\text{\rm red}}: \CM^\CF \to \R$ that
generates the thermodynamic equilibrium
 $R_{\CF;\CS}$ in $J^1\mathfrak{g}_{\CO_\CF}^* \cong J^1\R^n$ associated to the
 system $\CF$.  (See Section \ref{sec:generating} for the proofs.)
\end{enumerate}
We call the function $\CS^{\text{\rm red}}_\CF: \CM^\CF \to \R$ the
\emph{$\CF$-reduced information entropy}.
\end{thm}

\subsection{Volume, a non-local observable}

After these we revisit the previously well-studied thermodynamical models through
the eyes of our reduction. We first provide how to derive the volume variable $V$
as one of the variables of TPS from the first principle of SPS and KTPS.

In this regard, we would like to highlight that the volume is
\emph{not} a local observable in
statistical mechanics. It is not a \emph{dynamical} variable
in the level of SPS  but only a \emph{system} variable.
The volume variable appears to describe the state of particles confined is an
\emph{expandable container}. In the information theoretic
point of view, the volume is a parameter describing the
\emph{state of the container}, which is a state determined by
another \emph{event} triggered by an external 
force independent of the \emph{state of particles} 
contained in the container.

Because of this, we need to modify the definitions of SPS
and KTPS by incorporating
this difference between the volume variable and other local
observables in our kinetic theory framework, especially in our consideration 
of the probability distributions and relevant information entropy: The relevant KTPS will be now
$$
T^*\CP(\Gamma \times \R_+).
$$
See Section \ref{sec:volume} for the details. The volume variable
does not appear in SPS but it
will appear as a \emph{phase space variable}
 in the macroscopic TPS level .  It  then
carries a conjugate partner, which is nothing but the \emph{pressure}.

\subsection{Maxwell construction and phase transition}

The equilibrium state arising in Theorem \ref{thm:universality}
 is not necessarily $C^0$-holonomic
in terms of the given preferred variables.
We will illustrate by the van der Waals model that non-holonomicity indeed occurs in
a physical model which is responsible for the phase transition \cite{maxwell}.

Then in a thermodynamic equilibrium or nearby,  say in the system $\CF$ of ideal gas or
van der Waals gas equation, the relevant state variables are given by
$$
(U, S, T, P, V)
$$
where $U$ is energy, $S$ the thermodynamic entropy, $T$ temperature,
$P$ pressure and $V$ is the volume. The first law of equilibrium thermodynamics
is expressed as
$$
dS = \frac1T dU + \frac{P}{T} dV
$$
in the entropy representation. In our construction,
the relevant thermodynamic equilibrium is given by
the Legendrian submanifold $R$ generated by the aforementioned
(reduced) relative entropy function
$\CS_\CF^{\text{\rm red}}: M^\CF \to \R$ where
$$
\xymatrix{ \CM^\CF \ar[r]^{\CS_\CF^{\text{\rm red}}} \ar[d] & \R \\
\R^2
}
$$
 is a (infinite dimensional) symplectic fibration.

Since the aforementioned thermodynamic equilibrium may not be
globally holonomic in the preferred  variables, which are $P, \, T$
for the van der Waals model, the equilibrium equation cannot be globally
solved  in a direct way by simple quadratures in terms of the preferred variables.
An effort of correcting this deficiency
is precisely the celebrated Maxwell construction \cite{maxwell} in the equilibrium thermodynamics
for which the Maxwell equal-area law plays an important role.

It turns out that the Maxwell construction is  the same process as
that of  constructing the so called \emph{graph selector} out of non-holonomic
Lagrangian (resp. Legendrian) submanifold
in symplectic (resp.  in contact) geometry and in
dynamical systems. (See \cite{oh:jdg}, \cite{arnaud}, \cite{bernard-santos}
and \cite{amorim-oh-santos}
for example.) In Section \ref{sec:maxwell}, we relate the construction of a graph
selector with Gibbs free energy $G = U + PV -TS$ as its thermodynamic potential
to the relevant Maxwell construction and the phase transition
of the isotherms with respect to the pressure-volume diagram
(abbreviated as the $PV$-diagram). The outcome will be a
construction of
thermodynamic potential that is Lipschitz continuous but not differentiable
across the \emph{Maxwell pressure} at which Maxwell equal area law holds.
(See Theorem \ref{thm:maxwell-construction}.)

\subsection{Organization of the paper}

A brief organization of the paper is in order. Section \ref{sec:KTPS} introduces the
spaces SPS and KTPS and Section \ref{sec:Diff-symmetry} explains the underlying symmetry
group of KTPS. Section \ref{sec:SPS}--\ref{sec:time-evolution}
explains the Marsden-Weinstein
reduction of KTPS for the system of \emph{local observables}.
Section \ref{sec:SPS} explains what we mean by local observables
and observations and their invariance property on SPS.
Section \ref{sec:mesoscopic-reduction} carries out the
Marsden-Weinstein connection following our first motto and
construct the $\CF$-reduced KTPS $M^\CF$.
Section \ref{sec:reduced-entropy} constructs the $\CF$-reduced
information entropy $\CS_\CF^{\text{\rm red}}$ and shows that
it generates a thermodynamic equilibrium in the sense of the theory of
generating functions (or of the Morse family) See Appendix
\ref{sec:generating}). Then
Section \ref{sec:time-evolution} derives the simplest
 Bolzmann equation without collision terms
 associated to the total energy Hamiltonian $H$ on SPS as the internal
 mesoscopic dynamics of the obtained thermodynamic equilibrium.

Then in Section \ref{sec:volume}, we
explain a modification of the reduction process for the system involving non-local observables
of the volume variable of the gas models. In Section \ref{sec:ideal-gas}, we derive the
gas equations for the ideal gas. In Section \ref{sec:maxwell}
we identify the Maxwell construction as the process of selecting a graph-selector
in the sense of \cite{oh:jdg}, \cite{arnaud}, \cite{oh:C0-continuity}
and interpret the relevant phase transition for the van der Waals model in our framework. 
In Section \ref{sec:discussion} we have some
discussion on the dynamical perspective of our geometric study mostly of the kinematics 
of nonequilibrium thermodynamics, and
propose some future direction of research. In this regard, we
mention the relationship between our study and those given in
the recent interesting articles \cite{esen-grmela-pavelka-I},
\cite{esen-grmela-pavelka-II} by Esen-Grmela-Pavelka in particular.

Finally after the submission of the present paper, we have learned that
J.-M. Souriau \cite{souriau} provides a mathematical description of 
statistical mechanics and thermodynamics especially in the language of  
symplectic geometry and Poisson geometry. (See the article \cite{marle} of Marle for a survey of   \cite{souriau}.) In particular, Chapter IV of the book \cite{souriau} describes the 
thermodynamic equilibria based on the Lie group theory which shares some common spirit with 
our first motto ``Observation is a moment map''.  However the spirit of \cite{souriau}, \cite{marle}
is quite different from that of the present paper in that the main emphasis of the present paper lies 
in its systematic derivation of thermodynamics from statistical mechanics 
and in its contact geometric perspective motivated by the works of Jaynes \cite{Jaynes:InfoStat}
and Mrugala \cite{MNSS} in which the role of the (relative) information entropy 
is emphasized as fundamental.

\begin{center}
\bf Notations and Vocabularies
\end{center}

\begin{itemize}
\item SPS: statistical phase space $(\Gamma, \omega_\Gamma) = (T^*\R^{3N}, \Omega_0)$,
\item $\CP(\Gamma)$: the space of probability distributions,
\item $I= I(\rho)$: Relative information density (function),
\item $\CS:\CP(\Gamma) \to \R$ : relative information entropy, $\CS(\rho) = - \int_\Gamma
I(\rho) \rho$,
\item $\widetilde{\CS}:T^*\CP(\Gamma) \to \R$ : the lifted relative information entropy 
$\widetilde \CS: = \CS \circ \pi$ where $\pi: T^*\CP(\Gamma) \to \CP(\Gamma)$ is the natural projection,
\item KTPS: kinetic theory phase space $T^*(\CP(\Gamma)$,
\item $\CF = \{F_i\}_{i=1}^n$: observable system or statistical system,
\item $\CO_F$: the observation $\CO_F:\CP(\Gamma) \to \R$ of the observable $F$ 
defined by $\CO_F(\rho): = \int F\, \rho$,
\item $\widetilde{\CO}_F$: the lifted observation $\CO_F$ 
defined by $\widetilde \CO_F: = \CO_F \circ \pi$,
\item $\CM^\CF$: $\CF$-reduced KTPS,
\item $\CS^{\text{\rm red}}_\CF: \CM^\CF \to \R$: $\CF$-reduced entropy function,
\item TPS: thermodynamic phase space $J^1B$ for some open
subset $B \subset \R^n$,
\item $R_{\CF;\CS}$: the $\CF$ thermodynamic equilibrium.
\end{itemize}

\section{Statistical phase space and kinetic theory phase space}
\label{sec:KTPS}

Let $(\R^{2K} = T^*\R^K ,\Omega_0)$ be the Hamiltonian phase space equipped with
the standard symplectic form
$$
\omega_\Gamma = \sum_{i=1}^{K} dx^i \wedge dy^i.
$$
Our main interest will be the particle phase space
in gas dynamics, i.e., the classical \emph{statistical phase space}
 ($\SPS$) of many body systems where $K = 3N$ with $N$ being
 the number of particles in the system.

\begin{rem}
In our derivation, we will ignore the effect of collisions of the particles
in the story. To incorporate the effect of collisions as in the collisional kinetic theory
\cite{villani:kinetic}, we should consider the statistical phase space such as
$$
T^*(\R^{3N}\setminus \Delta_n) =: \Gamma
$$
where $\Delta_n$ is the big diagonal
$$
\Delta_N = \{({\bf q}_1, \ldots, {\bf q}_N) \in \R^{3N} \mid {\bf q}_i \in \R^3, \,  {\bf q}_i = {\bf q}_j, \, i \neq j\}.
$$
\end{rem}
\begin{rem}
The whole discussion can be applied to any smooth
manifold with a reference smooth measure, e.g., to any symplectic manifold equipped with
the Liouville measure on it, such as $\R^{3N} \setminus \Gamma$. However in the present paper, we will restrict ourselves to
the standard $\SPS$, except when we consider the nonlocal observable, the volume  in Section \ref{sec:ideal-gas} where we
employ a more general configuration space.
\end{rem}

We regard $\SPS$ $\Gamma$ as a symplectic manifold and denote it by
$$
(\Gamma, \omega_\Gamma), \quad \omega_\Gamma = \Omega_0.
$$
We may use $\Omega_0$ and $\omega_\Gamma$ as we feel more appropriate without distinction depending on circumstances.
Following the notation of \cite{MNSS}, we denote by $d\Gamma$ the density form of
the Lebesgue  measure of $\Gamma \subset \R^{2K}$ which coincides with the Liouville volume form
$$
\frac1{K!} \Omega_0^K =d \Gamma.
$$
Again we will also denote by $\nu_0$ the associated Liouville measure (or the Lebesgue measure). We also use $\nu_0$ and $d\Gamma$ interchangeably.

For a given signed density $D$, we consider the decomposition $D = D^+ - D^-$
where $D^\pm \geq 0$ are the positive and negative parts of $D$.

\begin{defn}  Write
$|D| = D^+ + D^-$ and define  $\CD(\Gamma)$ to be the set of densities $D$ having finite total mass
$$
M(D) := \int_\Gamma |D| < \infty.
$$
We denote by  $\CD^+(\Gamma)$ its subset consisting of nonnegative densities, and by $\CP(\Gamma)$
the subset consisting of probability densities, denoted by $\rho$,
whose total mass is 1, i.e., $M(\rho) = 1$.
\end{defn}

In this section, we set the state for applying the theory of symplectic reduction to an infinite dimensional
symplectic manifold, the cotangent bundle of $\CP(\Gamma)$,
$$
T^*\CP(\Gamma)
$$
as the Hamiltonian phase space. 
\begin{defn}[Kinetic theory phase space (KTPS)]
We call the cotangent bundle
$T^*\CP(\Gamma)$
equipped with its canonical symplectic structure the \emph{kinetic theory phase space} (KTPS).
\end{defn}

\begin{rem} \label{rem:bracket} A precise form of the canonical symplectic
structure does not play much role in the present paper except its general properties.
The most convenient way of writing down the (weak)
symplectic structure on $T^*\CP(\Gamma)$ as in the
Hamiltonian formalism of the general field theory may be in terms
of the associated Poisson bracket: For given real-valued functions
$\CF = \CF(\rho,\beta), \, \CG = \CG(\rho,\beta)$ on $T^*\CP(\Gamma)$, the canonical Poisson bracket between $\CF$ and $\CG$ are
given by
\begin{equation}\label{eq:bracket-FG}
\{\CF, \CG\}(\rho,\beta)
= \int_\Gamma \,
\left(\frac{\delta \CF}{\delta \rho}\frac{\delta \CG}{\delta \beta}
- \frac{\delta \CG}{\delta \rho} \frac{\delta \CF}{\delta \beta}\right)\,
\rho
\end{equation}
where  $\frac{\delta \CF}{\delta \rho}$ and $\frac{\delta \CF}{\delta \beta}$
stand for the variational derivative.
(See e.g., \cite{marsden-weinstein:vortices} for such a formula
and the precise meaning of the variational derivative.)
\end{rem}

Note that $\CD^+(\Gamma)$ is an open subset of a \emph{linear space}, and hence
its tangent space of $\CD^+(\Gamma)$ at $D$ can be identified with $\CD(\Gamma)$
which is naturally a vector space. On the other hand $\CP(\Gamma)$ is not a linear space but
a convex affine hypersurface
given by the integral constraint $\int \rho = 1$.
Because of this, we start our discussion with the cotangent bundle of $\CD^+(\Gamma)$ to 
study the symplectic geometry of $T^*\CP(\Gamma)$.

\begin{notation}\label{notation:dot-CD}
We adopt physicists' notation $\dot D$
to represent a tangent vector or the first variation of $D$ in $\CD^+(\Gamma)$.
Likewise we denote by $\dot \CD(\Gamma)$ the representative vector space $\CD(\Gamma)$
of the fiber of the tangent bundle $T\CD^+(\Gamma)$. Then we have
$$
\dot \CD(\Gamma) \cong T_D \CD^+(\Gamma) \cong \CD(\Gamma)
$$
for all $D \in \CD^+(\Gamma)$.
\end{notation}

Thanks to the presence of the Liouville volume form, any density $D$ has the form
$$
D = f\, d\Gamma
$$
for a nonnegative function $f: \Gamma \to \R^+$ with its $L^1$-norm finite,
provided $M(D) < \infty$.
The space $\CD^+(\Gamma)$ is a principal $\R_+$ homogeneous space that gives rise to
the following nonlinear exact sequence
\begin{equation}\label{eq:M}
0 \longrightarrow \CP(\Gamma) \longrightarrow \CD^+(\Gamma)  \stackrel{M}{\longrightarrow} \R_+ \to 0
\end{equation}
where the last map is given by the function
$D \to M(D)$.
By taking the logarithm $\log M(D)$, we get another exact sequence
\begin{equation}\label{eq:logM}
0 \longrightarrow \CP(\Gamma) \longrightarrow \CD^+(\Gamma)
\stackrel{\log M}{\longrightarrow} \R \to 0.
\end{equation}

We note that the $\R^+$-action on $\CD^+(\Gamma)$ is given by the map $(c,D) \mapsto c\, D$
which induces the action on $T\CD^+(\Gamma)$ given by
$$
(c, (D,\dot D)) \mapsto (cD, dR_c(\dot D)):
$$
Here we denote by
$R_c: \CD^+(\Gamma) \to \CD^+(\Gamma)$
the multiplication map by a constant $c > 0$.
We denote its dual action on $T^*\CD^+(\Gamma)$ by
$$
(c, (D,\alpha)) \mapsto (c\, D, (dR_c^{-1})^*\alpha)).
$$
In the level of tangent spaces, \eqref{eq:logM} also induces the exact sequence
\begin{equation}\label{eq:exact-sequence-D+}
0 \to T_\rho \CP(\Gamma) \to T_\rho \CD^+(\Gamma) \to
T_\rho \CD^+(\Gamma)/ T_\rho \CP(\Gamma)  \to 0.
\end{equation}
By setting the conormal space 
$$
\left(T_\rho \CD^+(\Gamma)/ T_\rho \CP(\Gamma)\right)^* \cong  (T_\rho \CP(\Gamma))^\perp
=: \nu_\rho^*\CP(\Gamma),
$$
its dual sequence is given by 
$$
0 \to \nu_\rho^*\CP(\Gamma) \to T_\rho^*\CD^+(\Gamma) \to T_\rho^*\CP(\Gamma) \to 0
$$
at each $\rho \in \CP(\Gamma)$. This sequence naturally splitts by the map 
\begin{equation}\label{eq:splitting-map}
T_\rho^* \CD^+(\Gamma)  \to \nu_\rho^* \CP(\Gamma);  \quad \beta \mapsto \beta - \beta(\rho) M
\end{equation}
when the mass map $M$ is regarded as a linear functional on 
$\dot D(\Gamma) = T_D \CD^+(\Gamma)$.

\section{$\Diff(\Gamma, \nu_0)$-symmetry and its moment map}
\label{sec:Diff-symmetry}

We first recall the following well-known fact.

\begin{lem} Any diffeomorphism of $\CP(\Gamma)$
canonically lifts to a symplectic diffeomorphism of $T^*\CP(\Gamma)$. 
In particular the induced action of the microscopic action of
$\Diff(\Gamma)$  on $\CP(\Gamma)$ by push-forward
is a symplectic action on $T^*\CP(\Gamma)$.
\end{lem}
\begin{proof}  Recall that for any diffeomorphism $\psi:N \to N$ of
a manifold $N$ (whether it is of finite dimension or not), the formula
$$
\widehat \psi(x, \beta): = (\psi(x), (d_x\psi^{-1})^*(\beta))
$$
provides the lifting to $T^*\CP(\Gamma)$ which is symplectic.
Furthermore this lifting is linear in fiber and in particular preserves the zero section. 
This shows that it canonically induces a symplectomorphism on $T^*\CP(\Gamma)$.
\end{proof}

We will use the following notation systematically.

\begin{notation}\label{notation:tilde-CG}
For any subgroup $\CG \subset \Diff(\CP(\Gamma))$,  we denote
by
$$
\widehat \CG \subset \Symp(T^* \CP(\Gamma))
$$
the associated lifting to $T^*\CP(\Gamma)$.
\end{notation}

Next we  describe the  moment map of this
symplectic action of $\Diff(\Gamma)$ on $T^*\CP(\Gamma)$ more explicitly.
For this purpose, we
note that $\Diff(\Gamma)$ action on $\CP(\Gamma)$ is nothing but the restriction of
the  natural action on the bigger space $\CD^+(\Gamma) \supset \CP(\Gamma)$:
For $\phi \in \Diff(\Gamma)$, it induces an action
by the push-forward
$
(\phi,D) \mapsto \phi_*D
$
on $\CD^+(\Gamma)$ and in turn by the action
\begin{equation}\label{eq:action-on-T*P}
(\phi,(D,\alpha)) \mapsto \left(\phi_*D, (d\phi^{-1})^*\alpha\right)
\end{equation}
on $T^*\CD^+(\Gamma) $.

Considering the identity component of $\Diff(\Gamma)$, we can identify its tangent vector
with a vector field $X \in {\mathfrak X}(\Gamma)$ after a right translation.
Denote by
$$
\CJ_{\Diff(\Gamma)} : T^*\CD^+(\Gamma) \to \mathfrak{X}(\Gamma)^*
$$
the universal moment map of  the action of $\Diff(\Gamma)$ on
$T^*\CD^+(\Gamma)$. (See \cite[Theorem 4.2.10]{abraham-marsden}.)
Then it is determined by its defining equation
\begin{equation}\label{eq:J-on-T*D}
\langle \CJ, X \rangle = \alpha (\widehat X)
\end{equation}
where $\widehat X$ is the vector field of $\CD^+(\Gamma)$
associated to $X$ under the lifting of the
linearized action of $\widehat{\Diff}(\Gamma)$.

\begin{exm} The tangent vector $\widehat X \in T_D \CD^+(\Gamma) \cong \CD(\Gamma)$ associated to
$X \in \mathfrak X(\Gamma)$ has the canonical pairing
\begin{equation}\label{eq:alpha-paired-X}
\alpha( \widehat X) = \int_\Gamma  \langle \alpha, X \rangle \, d\Gamma
\end{equation}
with each element $\alpha \in T_D^* \CD^+(\Gamma) \cong \CD^*(\Gamma)$,
by the canonical pairing between (vector-valued) functions and (vector-valued) distributions
on $\Gamma$.
\end{exm}

We note that as a subset of $T_\rho\CD^+(\Gamma)$ each tangent vector at $D$ of the $\Diff(\Gamma)$-orbit
has the form $\CL_X D$ and so the tangent space of a $\Diff(\Gamma)$-orbit thereat
is given by the subspace
\begin{equation}\label{eq:X-rho}
\widehat{\mathfrak X}_0(\Gamma;D) = \left\{ \widehat X \in \CD(\Gamma) \, \Big| \, \int \CL_X D = 0 \right\}.
\end{equation}

We now consider the hypersurface $M^{-1}(C)$ of constant mass $C>0$
consisting of positive densities $D$.
Recall that we have the natural isomorphism for the cotangent space,
\begin{equation}\label{eq:Trho*CP}
T_D^* M^{-1}(C) \cong \left\{\dot D \in \CD (\Gamma) \, \Big| \,  \int_\Gamma \dot D = 0 \right\}^\perp \subset \CD(\Gamma)^*.
\end{equation}
With this canonical identification, we can summarize the above discussion into
the following explicit formula for the  moment map
of the canonical action of
$\Diff(\CP(\Gamma))$ on $T^*M^{-1}(C)$ as follows.

\begin{thm}[Universal moment map]
\label{thm:momentmap-formula} 
Consider the hypersurface $M^{-1}(C) \subset \CD^+(\Gamma)$.
The moment map of the action of $\Diff(\Gamma,\nu_0)$
$$
\CJ_{\Diff(\Gamma,\nu_0)}: T^* M^{-1}(C) \to
\widehat{\mathfrak X}_0(\Gamma;D)^*
$$
is expressed as
\begin{equation}\label{eq:Diff-momentmap}
\CJ_{\Diff(\Gamma,\nu_0)}(D, \beta) =
 \left(\beta - \frac{\beta(D)}{C} M\right)
\Big |_{{\mathfrak X}_0(\Gamma;D)}
\end{equation}
under the identifications \eqref{eq:X-rho} and \eqref{eq:Trho*CP}.
\end{thm}

For example, this theorem applies to the probability space $\CP(\Gamma) = M^{-1}(1)$.

\begin{cor}
\label{cor:momentmap-formula} The moment map of the action of $\Diff(\Gamma,\nu_0)$
on $T^*\CP(\Gamma)$  
$$
\CJ_{\Diff(\Gamma,\nu_0)}: T^*\CP(\Gamma) \to
\widehat{\mathfrak X}_0(\Gamma;\rho)^*
$$
is given by
\begin{equation}\label{eq:Diff-momentmap}
\CJ_{\Diff(\Gamma,\nu_0)}(\rho, \beta) =
 \left(\beta - \beta(\rho) M\right)
\Big |_{{\mathfrak X}_0(\Gamma;\rho)}
\end{equation}
under the identifications \eqref{eq:X-rho} and \eqref{eq:Trho*CP}.
\end{cor}
We note that the map coincides with the splitting map \eqref{eq:splitting-map}
when $\rho$ is fixed.

\section{Statistical phase space and systems}
\label{sec:SPS}
Let us equip $T^*\R^K$  with the canonical symplectic form
$$
\omega_\Gamma = \sum_{i=1}^K dq_i \wedge dp_i, \quad K = 3N.
$$
Since any  symplectic diffeomorphism satisfies
$\phi^*\omega_\Gamma = \omega_\Gamma$, it also preserves the Liouville measure
$\phi^*d\Gamma = d\Gamma$.

We adopt the following terminologies coming from statistical mechanics.

\begin{defn}[Observables and observations]\label{defn:observation}
We call a function $F: \Gamma \to \R$ a \emph{local observable} and
a collection of functions
$$
\CF = \{F_1, \ldots, F_n\}
$$
a \emph{local observable system}.
\begin{enumerate}
\item Each  probability distribution $\rho \in \CP(\Gamma)$
with respect to which $F$ is an $L^1(\rho)$-function defines
the average
\begin{equation}\label{eq:CO}
\CO_F(\rho) = \langle F \rangle_\rho: = \int F\, \rho.
\end{equation}
\item
We call this macroscopic average of $F$  an \emph{observation}
of the local observable $F$ relative to $\rho$.
\item
We define a \emph{collective observation} to be the map
$
\CO_\CF: \CP(\Gamma) \to \R^n
$
given by
\begin{equation}\label{eq:collective-O}
\CO_\CF = (\CO_{F_1},\ldots, \CO_{F_n}).
\end{equation}
\end{enumerate}
\end{defn}

We will also denote the value of $\CO_{F_i}$
at $\rho$ by the standard notation $\CO_{\CF_i}(\rho)
= \langle F_i \rangle_\rho$ in statistical mechanics, when we feel
its use more suggestive.

We emphasize that the relative information entropy $\CS$ is the observation
associated to an observable, the \emph{relative information density},
 which is \emph{intrinsic} and \emph{universal}
in the sense that $\CS$ depends only on the probability density $\rho$ independent of other local observables.

\begin{defn}[Relative information entropy]
The relative information entropy  $\CS(\rho)$ 
(aka \emph{Kullback-Leibler divergence} $D_{KL}(\rho\|\nu_0)$) is defined to be
$$
\CS(\rho): = - D_{KL}(\rho\|\nu_0) = - \int \rho \log \frac{\del \rho}{\del \nu_0},\quad \nu_0 = d\Gamma.
$$
\end{defn}

The following is obvious. However we would like to remark that consideration of
\emph{relative information entropy} $- D_{KL}(\rho\|\nu_0)$ in the presence of the background
Liouville measure $\nu_0 = d\Gamma$, \emph{not} Shannon's original continuous entropy,
is crucial to have this invariance property.

\begin{prop}\label{prop:invariance-entropy}
The (relative) information entropy is invariant under the action of
$\nu_0$-measure preserving diffeomorphisms of $\SPS$ $\Gamma$. In particular it is preserved
under the action by symplectic diffeomorphisms on $\Gamma$.
\end{prop}
\begin{proof} 
Recall the action of a diffeomorphism $\phi: \Gamma \to \Gamma$  is given by
$(\phi, \rho) \mapsto \phi_*\rho$
and by the definition of the Radon-Nikodym derivative that
$$
\frac{\del (\phi_*\rho)}{\del \nu_0} = \frac{\del \rho}{\del \nu_0} \circ \phi^{-1}.
$$
Therefore we compute
\beastar
D_{KL}(\phi_*\rho\|\nu_0) & = & 
\int_\Gamma \phi_*\rho \log \left(\frac{\del (\phi_*\rho)}{\del \nu_0}\right)
 \\
& =& \int_\Gamma  \left(\left(\frac{\del \rho}{\del \nu_0} \circ \phi^{-1}\right)\, \nu_0\right)
\log \left(\frac{\del \rho}{\del \nu_0} \circ \phi^{-1}\right) \\
& =& \int_\Gamma  \left(\left(\frac{\del \rho}{\del \nu_0} \circ \phi^{-1}\right)\, \phi_*\nu_0\right)
\log \left(\frac{\del \rho}{\del \nu_0} \circ \phi^{-1}\right) \\
& = & \int_\Gamma \left ( \frac{\del \rho}{\del \nu_0} \, \nu_0\right)
\log \left(\frac{\del \rho}{\del \nu_0}\right) = D_{KL}(\rho\| \nu_0)
\eeastar
where the third equality follows since $\phi$ is assumed to preserve the $\nu_0$-measure. 
This finishes the proof.
\end{proof}

\section{Mesoscopic reduction of KTPS: the first motto}
\label{sec:mesoscopic-reduction}

In this section, we perform the first stage of thermodynamic reduction
by decomposing KTPS with those of iso-data over the collective observation data by applying Marsden-Weinstein reduction.

\begin{rem}
It appears that this reduction process is related to the theory of \emph{coarse graining}
in statistical mechanics. Coarse-graining is the process of grouping together of a large number of
small entities into some larger size ensemble with important characteristics of small
entities intact and then analysing this new system. Any description in between the microscopic level of
statistical mechanics and the macroscopic of thermodynamics is called a \emph{mesoscopic level}.
(See \cite{espanol} for example.)
\end{rem}

We start with the following obvious lemma.

\begin{lem} \label{lem:Diff-invariance}
Consider any system $\CF = \{F_1, \cdots, F_n\}$ of local observables.
The collective observation $\{\CO_{F_i}\}$ is invariant under the action of the group
$$
\Diff(\Gamma, \nu_0).
$$
\end{lem}
\begin{proof} This follows by the chain rule this time more easily using
the measure-preserving hypothesis of the action as in the proof of Proposition
\ref{prop:invariance-entropy}.
\end{proof}

Suppose we are given a local observable system $\CF$, and
consider its collective  observation function
$\{\CO_{F_i}\} _{i=1}^n$.

\begin{defn}\label{defn:tilde}
We then lift each function $\CO_{F_i}$ to $\widetilde \CO_{F_i} = \CO_{F_i} \circ \pi$ 
defined on $T^*\CP(\Gamma)$ i.e.,
$$
\widetilde \CO_{F_i}(\rho,\beta) : = \CO_{F_i}(\rho)
$$
and call the \emph{ lifted observation} of $F_i$.
Similarly we write $\widetilde{\CS} = \CS \circ \pi$ and call it the \emph{lifted relative information entropy}.
\end{defn}

Then the functions $\{\widetilde{\CO}_{F_i}\} _{i=1}^n$
Poisson commute with respect to the
canonical symplectic structure on $T^*\CP(\Gamma)$
given by \eqref{eq:bracket-FG} and so
their collective Hamiltonian flows on $T^*\CP(\Gamma)$
defines a Hamiltonian action of the abelian group $\R^n$.
Each element of these Hamiltonian flows is derived from
a  $\nu_0$-measure preserving  (local) diffeomorphism group symmetries of $\Gamma$.

\begin{defn}\label{defn:G-CF}
For the given set of observables $\CF=\{F_1,\dots,F_n\}$ on SPS
$\Gamma$, we consider the Lie algebra 
$\mathfrak g_{\CO_\CF} \subset {\mathfrak X}^{\text{\rm symp}}(T^*\CP(\Gamma))$
generated by the commuting Hamiltonian vector fields
\begin{equation}\label{eq:frakgCF}
\left \{X_{\widetilde{\CO}_{F_1}}, \cdots, X_{\widetilde{\CO}_{F_n}}\right\}
\end{equation}
on $T^*\CP(\Gamma)$.
\end{defn}

We recall the inclusions
$$
\CP(\Gamma) \subset \CD^+(\Gamma)  \subset \CD(\Gamma)
$$
where the last space is a vector space, the second is an open subset thereof and
the first is a hypersurface of the second. Furthermore the observation function
can be linearly extended to $\CD(\Gamma)$ by the same formula
$$
\int_\Gamma F\, D =: \mathcal O_F(D), \quad D \in \CD(\Gamma).
$$
The splitting \eqref{eq:splitting-map} induces the decomposition
$$
T_\rho \CD(\Gamma) = T_\rho \CP(\Gamma) \oplus \R \langle \rho \rangle
$$
given by the explicit formula
\begin{equation}\label{eq:decomposition}
\dot D \mapsto \left (\dot D - \left(\int_\Gamma \dot D\right) \, \rho,
 \left(\int_\Gamma \dot D \right)\, \rho \right).
\end{equation}

Now we derive the formula for the associated moment map
$$
\CJ_{\CF}: T^*\CP(\Gamma) \to \mathfrak{g}_{\CO_\CF}^*.
$$
Recall that the assignment $F \mapsto X_{\CO_F}$ defines
a map
$$
C^\infty(\Gamma) \to {\mathfrak X}^{\text{\rm symp}}(T^*\CP(\Gamma)).
$$
This in turn induces the dual map
$$
{\mathfrak X}^{\text{\rm symp}}(T^*\CP(\Gamma))^*
\to (C^\infty(\Gamma))^* \cong \CD(\Gamma).
$$

The following proposition is a crucial link between the statistical system
and the thermodynamic system. We phrase this proposition as our first motto:

\begin{center}
\bf ``Observation is a moment map.''
\end{center}

\begin{prop}
The moment map
$$
\CJ_{\CF} : T^*\CP(\Gamma) \to \mathfrak g_{\CO_\CF}^*
$$
of the action of $\CG_{\CF}$ on
$T^*\CP(\Gamma)$ is characterized by the formula
\begin{equation}\label{eq:J-FS}
\left\langle \CJ_{\CF}(\rho,\beta), \frac{\del}{\del F_i} \right\rangle =
\CO_{F_i}(\rho) - \beta(\rho)
\end{equation}
for all $i$.
\end{prop}
\begin{proof} We first note that the function $\CO_\CF: \CP(\Gamma) \to \R$
is the natural restriction of the linear map defined on the bigger space $\CD(\Gamma)$.
We denote this latter map by
$$
\widehat {\CO}_\CF: \CD(\Gamma) \to \R.
$$
For the latter map which is a linear function on the vector space $\CD(\Gamma)$,
we apply the following standard lemma.

\begin{lem} Let $V$ be a vector space and consider its cotangent bundle
$\pi: T^*V \to V$. Let $f: V \to \R$ be a linear function and
consider the function $f\circ \pi : V \oplus V^* \cong T^*V \to \R$
where $\pi: T^*V \to V$ the cotangent projection. Denote by
$\phi_{f\circ \pi}^t$ be the one-parameter subgroup of
the Hamiltonian flow of $f \circ \pi$ regarded as the action of
the Lie group $\R$ acting on $T^*V$. Then its moment map
$J: T^*V \to \R^*$ is characterized by
$$
\left\langle J(v,\beta), \frac{\del}{\del t} \right \rangle  = f \circ
\pi(v,\beta) ( = (f(v)).
$$
\end{lem}
\begin{proof} By the defining condition of the moment map,
$$
\left\langle J(v,\beta), \frac{\del}{\del t} \right \rangle
$$
is the Hamiltonian generating the flow $\phi_{f\circ \pi}^t$ which is
obviously given by $f\circ \pi$.
\end{proof}

Therefore the Hamiltonian generating
the flow associated to the observable $F_i$ is precisely the function
$\widetilde{\CO}_{F_i} \circ \pi_{T^*\CD(\Gamma)}$.
In the current case, we are considering the $n$-dimensional Lie algebra
whose generating vector fields are those whose flows are the linear translations
$$
(\rho,\beta) \mapsto (\rho, \beta + d\CO_{F_i})
$$
respectively for each $i$ by the definition of the Lie algebra
action of $\CG_\CF$.

Since the Hamiltonian vector field $X_{\CO_{F_i}}$ on $T^*\CD(\Gamma)$ is
vertical, the flow preserves the subset
$$
T^*\CD(\Gamma)|_{\CP(\Gamma)} \subset T^*\CD(\Gamma).
$$
We also have the natural projection map
$$
T^*\CD(\Gamma)|_{\CP(\Gamma)} \to T^*\CP(\Gamma).
$$
Applying \eqref{eq:Diff-momentmap},
we derive
$$
\left\langle \CJ_\CF(\rho, \beta), \frac{\del}{\del F_i} \right\rangle  \mapsto
\left(\CO_{F_i}(\rho) - \frac{\beta(\rho)}{M(\rho)} M(\rho) \right)
\Big |_{{\mathfrak X}_0(\Gamma;\rho)}=  \CO_{F_i}(\rho)- \beta(\rho)
$$
 using $M(\rho) = 1$ under the identification \eqref{eq:Trho*CP}.
 This finishes the proof.
\end{proof}

We summarize the above discussion into the following commutative diagram
\begin{equation}\label{eq:generating-diagram}
\xymatrix{ & {T^*\CP(\Gamma)}  \ar[dl]_{\CJ_\CF} \ar[d]^{\CJ_{\CP(\Gamma)}}\ar[r]^>>>>>{\widehat \CS}& \R
\\
{\mathfrak{g}_{\CO_\CF}^*}  &\ar[l] {\mathfrak{X}(\Gamma)^*} & }
\end{equation}
\begin{defn}[Observation data set] We denote by $B_\CF$ the image of the moment map $\CJ_\CF$
and by $B^{\circ}_\CF$ the set of its regular values. We call $B_\CF$ the \emph{observation data set}.
\end{defn}
Then we have decomposition
$$
T^*\CP(\Gamma) = \bigcup_{\mu \in B_\CF} \{\mu\} \times \CJ_\CF^{-1}(\mu)
$$
and
$$
T^*\CP(\Gamma)/\CG_\CF = \bigcup_{\mu \in B_\CF} \{\mu\} \times \CJ_\CF^{-1}(\mu)/\CG_\CF.
$$
We attract readers' attention that the coadjoint isotropy group
of $\CG_\CF$  is the full group
\begin{equation}\label{eq:isotropy-group}
\CG_{\CF,\mu} = \CG_\CF
\end{equation}
for all $\mu$ and hence the reduced space at $\mu$
 \cite{marsden-weinstein} becomes
$$
\CJ_\CF^{-1}(\mu) /\CG_\CF = : \CM^{\CF}_{\mu}
$$
for all regular values of $\mu$. Obviously $\CJ_\CF^{-1}(\mu) = \emptyset$ unless
$\mu \in B_\CF$.

We summarize the above discussion into
\begin{cor} Let $ \mu = (\mu_1, \cdots, \mu_n)$ be a collective observation of
$\CF = \{F_1, \cdots, F_n\}$ and assume $\mu \in \mathfrak g_{\CO_\CF}^*$ is a regular value
of the moment map $\CJ_{\CF}$. Then the reduced space denoted by
\begin{equation}\label{eq:MmuF}
\CM_{\mu}^{\CF} : = \CJ_\CF^{-1}(\mu) /\CG_\CF\end{equation}
is a (infinite dimensional) symplectic manifold.
\end{cor}

The projection $\CM^{\CF} \to B^{\circ}_\CF \subset {\mathfrak g}_{\CO_\CF}^*$ forms 
a \emph{symplectic fiber bundle} over
$B^{\circ}_\CF\subset \mathfrak g_{\CO_\CF}^*$.
\begin{defn} 
We call the union
$$
\CM^{\CF} : = \bigsqcup_{\mu \in B^{\circ}_\CF}  \{\mu\} \times \CM^{\CF}_{\mu};
\quad \mu: = (\mu_1,\cdots, \mu_n)
$$
the \emph{$\CF$-reduced kinetic theory phase space} ($\CF$-reduced $\KTPS$)
associated to $\CF$,
where $\CJ = \CJ_\CF$ is the moment map associated to the symmetry group
generated by the induced Hamiltonian flows on $T^*\CP(\Gamma)$.
\end{defn}

The following is obvious.

\begin{cor} Suppose that the set $\{X_{\widetilde{\CO}_{F_i}}\}_{i=1}^n$ is linearly
independent. Then the map $\R^n \to \mathfrak g_{\CO_\CF}^*$ defined by
$$
(c_1,\cdots, c_n) \mapsto \sum_{i=1} c_i \widetilde{\CO}_{F_i}
$$
is an isomorphism.
\end{cor}
We would like to mention that the hypothesis of this corollary is a very weak one.
For example,
the hypothesis holds for any local observable system that has a point $x \in \Gamma$
at which the differentials
$$
\{dF_1(x), \cdots, dF_n(x) \}
$$
are linearly independent.

\section{Reduced entropy as a generating function of thermodynamic equilibrium: the second motto}
\label{sec:reduced-entropy}

We have shown before that the relative information entropy
$\CS:\CP(\Gamma) \rightarrow \mathbb{R}$ is invariant
under the action of $\Diff(\Gamma,\nu_0)$ in Proposition
\ref{prop:invariance-entropy}. Therefore the lifted relative information entropy $\widetilde{\CS}$
is also invariant under the lifted action $\widehat{\Diff}(\Gamma,\nu_0)$.

The following will enable us to derive the folklore
that \emph{the  thermodynamic entropy  is derived as the reduction
of the information entropy}.

\begin{cor} [$\CF$-reduced  entropy]\label{cor:universality}
The lifted relative information entropy $\widetilde \CS = \CS \circ \pi: T^*\CP^(\Gamma) \to \R$
 is universally reduced to a well-defined function
 $$
 \CS_\CF^{\text{\rm red}}: \CM^{\CF} \to \R.
 $$
 We call $\CS^{\text{\rm red}}_\CF$ the \emph{$\CF$-reduced entropy function}.
\end{cor}
In fact the above definition can be put into the general construction of
\emph{Legendrian generating function} as follows. (See Appendix \ref{sec:generating} for the definition.)

The procedure of finding a critical point of $\CS_\CF^{\text{\rm red}}$
can be decomposed into the two steps. We first choose
an Ehresmann connection
$$
T\CM^\CF = T^v \CM^\CF \oplus T^h \CM^\CF
$$
of the fibration  $\pi_\CF: \CM^\CF \to \mathfrak{g}^*_{\CO_\CF}$,
and decompose the differential $d\CS_\CF^{\text{\rm red}}(\rho)$
into the vertical and the horizontal components 
$$
d\CS_\CF^{\text{\rm red}}(\rho) = d^v\CS_\CF^{\text{\rm red}}(\rho)
+ D^h\CS_\CF^{\text{\rm red}}(\rho).
$$
Recall that the vertical differential is canonically defined but the
horizontal differential of $\CS_\CF^{\text{\rm red}}$ needs the use of connection. In general the horizontal component $D^h\CS(\rho)$ depends on the
connection but it will be independent thereof at the vertical critical
point $\rho$ where $d^v\CS(\rho) = 0$.

Then we consider the following diagram induced by the diagram \eqref{eq:generating-diagram}
\begin{equation}\label{eq:generating-diagram2}
\xymatrix{{\CM^\CF}  \ar[d]^{\pi_\CF} \ar[r]^>>>>>{\CS^{\text{\rm red}}_\CF}& \R
\\
{\mathfrak{g}_{\CO_\CF}^*}  }
\end{equation}
This leads us to our second motto:

\begin{center}
\bf ``Relative information entropy is \\
the generating function of thermodynamic equilibria."
\end{center}

\begin{defn}[Covariant thermodynamic equilibrium]
We call the Legendrian submanifold defined by
\bea\label{eq:RCF}
R_{\CF;\CS} : = \left\{\left(\mu, D^h\CS_\CF^{\text{\rm red}}(\rho,[\beta]),
\CS_\CF^{\text{\rm red}} (\rho,[\beta])\right )\in J^1(\mathfrak{g}_{\CO_\CF}^*)\,
\big|\,  d^v\CS_\CF^{\text{\rm red}}(\rho,[\beta]) = 0, \right. \nonumber \\
\left. \mu = \CJ_\CF(\rho,[\beta])   \right\}
\eea
the \emph{covariant thermodynamic equilibrium of statistical system $\CF$.} 
We regard the diagram \eqref{eq:generating-diagram2}
as a  generating function of the Legendrian submanifold
 $R_{\CF;\CS}$.
\end{defn}

\begin{prop} $R_{\CF;\CS}$ is a Legendrian submanifold of
the 1-jet bundle $J^1(\mathfrak{g}_{\CO_\CF}^*)$
for the contact form
$$
\lambda = dw - \sum_{i=1}^n q_i dp_i
$$
whose potential function is given by $\CS_\CF^{\text{\rm red}}$.
\end{prop}
\begin{proof} We will show that the subset
$$
\CL_{\CF;\CS}: = \left\{(\mu, D^h\CS_\CF^{\text{\rm red}}(\rho,[\beta]))
(\rho,[\beta]))\in T^* \mathfrak{g}_{\CO_\CF}^*\,
\big|\, d^v\CS_\CF^{\text{\rm red}}(\rho,[\beta]) = 0\right\}
$$
is a Lagrangian submanifold of $T^*\mathfrak{g}_{\CO_\CF}^*$. We denote by
$\iota: \CL_{\CF;\CS} \to T^*(\mathfrak{g}_{\CO_\CF}^*)$ the canonical inclusion. Then it
 satisfies
$$
dg = \iota^*\theta, \quad \theta = \sum_{i=1}^n q_i dp_i
$$
with $g = \iota^*\CS_{\CF}^{\text{\rm red}}$. The latter is because by definition  we have
$$
d\CS_{\CF}^{\text{\rm red}} = D^h\CS_\CF^{\text{\rm red}}
$$
on $\CL_{\CF;\CS}$. In particular, we have $0 = \iota^*d\theta = - \iota^*\omega_0$ which shows that
$\CL_{\CF;\CS}$ is a Lagrangian submanifold whose Liouville primitive is given by
the function $g$.
\end{proof}

This construction is the coordinate free construction of
the local coordinate description
of a thermodynamic equilibrium given in the literature (e.g., in \cite{MNSS}): The latter construction
has been carried out by
the Lagrangian multiplier method for the entropy maximum principle.
For the readers' convenience, we summarize this local
coordinate description of thermodynamic equilibrium
in Section \ref{sec:review}.

Non-holonomicity of $R_{\CF;\CS}$ arises from the last operation of `taking the
horizontal projection'. 

\section{Mesoscopic dynamics and conservation equation}
\label{sec:time-evolution}

In the previous section, we have constructed thermodynamic equilibrium state as the
Legendrian submanifold generated by the $\CF$-reduced entropy function. In this section,
we describe the mesoscopic reduction of the kinetic equation on KTPS $T^*\CP(\Gamma)$ generated 
by the Hamiltonian flow on SPS of the given Hamiltonian $H$.

\begin{rem}[Physical Hamiltonians]
In the physical thermodynamical system, the Hamiltonian
 is nothing but the total energy, i.e., the kinetic energy plus the potential energy
\begin{equation}\label{eq:physical-H}
H = \frac{m}{2} \sum_{i=1}^N |\vec P_i|^2 + \sum_{i\neq j} U_{ij}(Q_i - Q_j) + \sum_{\ell =1}^N U(Q_\ell).
\end{equation}
\end{rem}

\begin{lem} The (infinite dimensional) Hamiltonian flow 
$$
\phi_{\widetilde{\CO}_H}^t: T^*\CP(\Gamma) \to T^*\CP(\Gamma)
$$
on $KTPS$ preserves the lifted observations $\widetilde{\CO}_{F_i}$ and the
lifted relative information entropy $\widetilde{\CS}$. We write $\phi_{\widetilde{\CO}_H}^t=: \Phi_H^t$.
\end{lem}
\begin{proof} We note that $\phi_H^t$ preserves the Liouville measure $\nu_0$. Therefore
by the change of variable formula, we have
$$
\widetilde{\CO}_{F_i} \circ \Phi_H^t = \widetilde{\CO}_{F_i}, \quad \widetilde{\CS} \circ \Phi_H^t 
= \widetilde{\CS}.
$$
We recall the fact that all  $\widetilde{\CO}_{F_i}$ and $\widetilde \CS$ are 
the lifts of functions defined originally on $\CP(\Gamma)$. Therefore they Poisson-commute
as functions on $T^*\CP(\Gamma)$ by Proposition \ref{prop:invariance-entropy}
and Lemma \ref{lem:Diff-invariance}.  Then the lemma follows.
\end{proof}

Therefore the flow $\Phi_H^t$ descends to $\CM^\CF$.
We denote this flow by
$$
\Phi_{H;\text{\rm red}}^{\CF,t}: \CM^\CF \to \CM^\CF.
$$
Recall that the flow $\Phi_H^t$  on $T^*\CP(\Gamma)$ preserve the level sets of
$\widetilde{\CO}_{F_i}$ and commute with the flows of $X_{\widetilde{\CO}_{F_i}}$ for all $i = 1, \cdots, k$. 
We also recall that the flow on $\CM_\CF$ is the restriction of the quotient flow of
$T^*\CP(\Gamma)/\CG_\CF$.  In particular the reduced flow preserves the value of the
$\CF$-reduced entropy $\CS_\CF^{\text{\rm red}}$. 

Even in the thermodynamic equilibrium $R_{\CF;\CS}$, the original
kinetic system still induces the internal flow thereon induced by the time-evolution
 of the probability distribution in SPS (or in KTPS) which
 may  undergo nontrivial time-evolution. 
 More specifically the flow of the probability distribution
 satisfies the reduced equation of the kinetic equation
\begin{equation}\label{eq:drhodt}
\frac{\del \rho}{\del t} = - \CL_{X_H} \rho, \quad
\rho(0) = \rho_0.
\end{equation}
By writing $\rho = f\, \nu_0$ with $f = \frac{\del \rho}{\del \nu_0}$,
the latter equation  is equivalent to the conservation equation
\begin{equation}\label{eq:frho}
\frac{\del f}{\del t} + \{f, H\} =0, \quad f(0) = f_{\rho_0}.
\end{equation}
This equation is the simplest form of
the Boltzmann equation that does not incorporate the contribution from collision of particles. 
We refer readers to Section \ref{sec:discussion} for further discussion
on the more general form of kinetic equations.

By differentiating the macroscopic entropy
 $\CS(\rho_t)$ in time $t$, we get
 \begin{prop}\label{prop:dCSdt} Let $H = H(t,x)$ be the Hamiltonian of SPS and let
 $\rho_t= (\phi_H^t)_*\rho_0$ for any $\rho_0 \in \CP(\Gamma)$. Then we have
 \begin{equation}\label{eq:dfrhodt}
 \frac{\del}{\del t}(\CS(\rho_t)) =  \int_\Gamma \{f_t \log f_t, H\}\, d\Gamma.
 \end{equation}
 \end{prop}
 \begin{proof} We compute
 \beastar
 \frac{\del}{\del t}(\CS(\rho_t))
 & = & - \int_\Gamma \frac{\del}{\del t}(f_t \log f_t) \, d\Gamma 
= - \int_\Gamma (\log f_t +1)\frac{\del f_t}{\del t}\, d\Gamma \\
& = & - \int_\Gamma (\log f_t +1) \frac{\del f_t}{\del t}\, d\Gamma \\
 & = & \int_\Gamma (\log f_t +1) \{f_t, H\}\, d\Gamma 
= \int_\Gamma \{f_t \log f_t, H\}\, d\Gamma
 \eeastar
where the last equality follows from the derivation property of bracketing $\{\cdot, H \}$.
This finishes the proof.
\end{proof}

We now rewrite 
$$
\{f_t \log f_t, H\}\, d\Gamma
= \CL_{X_H}(f_t \log f_t)\, d\Gamma = \CL_{X_H}(f_t \log f_t\, d\Gamma)
$$
where the penultimate equality holds as $\CL_{X_H} (d\Gamma) = 0$ by Liouville's lemma.

\begin{cor}\label{cor:CSt-constant} 
We put $s_t: = s(\rho_t) = f_t \log f_t$ and 
$$
S^{6N-1}(R): = \{(Q,P) \mid |Q|^2 + |P|^2 = R^2\}.
$$
Suppose the pair $H$ and $f$ satisfy the asymptotic boundary condition
\begin{equation}\label{eq:asymp-condition}
\lim_{R \to \infty} \int_{S^{6N -1}(R)} s_t (X_H \rfloor d\Gamma) = 0.
\end{equation}
Then the entropy 
$\CS(\rho_t)$ remains constant in time $t$ for any solution $f_t$ of \eqref{eq:frho}.
\end{cor}
\begin{proof} By Cartan's formula, we have
$$
\CL_{X_H}(f_t \log f_t\, d\Gamma) = d((s_t X_H) \rfloor d\Gamma).
$$
Therefore 
\begin{equation}\label{eq:Stokes'formula}
\int_{B^{6N}(R)} d((s_t X_H) \rfloor d\Gamma) = \int_{S^{6N -1}(R)} (s_t X_H) \rfloor d\Gamma
= \int_{S^{6N -1}(R)} s_t (X_H \rfloor d\Gamma)
\end{equation}
by Stokes' formula. The the corollary follows by taking the limit as $R \to \infty$.
\end{proof}

\begin{rem} Strictly speaking, we need suitable
integrability and regularity of $f_t \log f_t$ to make the above proposition and corollary 
completely rigorous so that we can switch the order of integration and differentiation in time
in the proof of Proposition \ref{prop:dCSdt}, and the limits of the far left and the far right
hand sides of \eqref{eq:Stokes'formula} exist in the proof of Corollary \ref{cor:CSt-constant}.
Since identifying the optimal class of probability densities and the physical Hamiltonian $H$ that satisfy 
\eqref{eq:asymp-condition} is beyond the main theme of the present work, we do not elaborate here but
postpone further detailed study of such conditions and their kinetic consequences elsewhere.
\end{rem}

\section{Non-holonomicity of thermodynamic equilibrium and phase transition}
\label{sec:review}

Our global derivation of thermodynamic equilibrium given in the previous section
indicates that
the associated Legendrian submanifold is \emph{not necessarily holonomic}.
We will illustrate by the van der Waals model that non-holonomicity indeed occurs in
a physical model, which is responsible for the phase transition. Then the well-known
Maxwell's construction is an attempt to overcome this non-holonomicity in computing
the relevant \emph{single-valued} thermodynamic potential function.

We first review the exposition of \cite{Jaynes:InfoStat}, \cite{MNSS}
on the thermodynamic equilibrium
in terms of the variables $(q^1,\cdots, q^n)$ which is a coordinate system of
$B_{\CF;\CS} \subset \R^n$.
We will closely follow the exposition of \cite{MNSS}, especially try to be consistent
therewith in the usage of letters for the relevant variables except the change of
$x^i$'s by $q^i$'s.

We first give the definition of TPS associated to the observable $\CF = \{F_1, \ldots, F_n\}$.
\begin{defn}
A \emph{thermodynamic phase space} $\CT$ is an open subset of $\mathbb{R}^{2n+1}$
 equipped with the contact structure of the 1-jet bundle
naturally associated to a statistical mechanical system
 of observables $F_1, \cdots, F_n$ on SPS.
\end{defn}

In the contact geometric formulation of TPS, it is described by a coordinate system
$(q^1,\dots,q^n;p_1,\cdots,p_n;z)$
where $q^1,\cdots,q^n$ are called \emph{configuration variables}, and the relevant contact
structure is defined by the kernel of the contact form of the type
 $$
 dz - \sum_i q^i dp_i.
 $$
 In this formulation, the contact structure on TPS is related to the first law of thermodynamics, and Legendrian submanifold is related to an equilibrium state.

Note that $\nu_0$ is the Lebesgue measure and we can express
any nonnegative density $D$ as $D = f \, \nu_0$ 
for a nonnegative $L^1$-function $f$ on SPS with
$f =\frac{\del D}{\del \nu_0}$.

We now consider the system with $n$-observables, $F_1,\dots,F_n$.
Then, we want to find a probability density
$\rho \in \mathcal{P}(\Gamma)$ that maximizes the
relative entropy functional
$$
\CS(\rho)=-\int_\Gamma \rho \log \left(\frac{\del \rho}{\del \nu_0}\right).
$$
under the constraints
\begin{equation}\label{eq:pi=intFi}
 q^i  = \frac{\int_\Gamma F_i D}{\int_\Gamma D} =  \frac{\langle F_i \rangle_D}{M(D)}
\end{equation}
for given explicit observations $(q^1,\cdots,q^n)$ \cite{Jaynes:InfoStat}.
We note that while $\CD^+(\Gamma)$ is an open subset of a \emph{linear space},
$\CP(\Gamma)$ is not which is an affine hypersurface given by the integral constraint
\begin{equation}\label{eq:intrho=1}
\int \rho = 1.
\end{equation}
Instead of solving this maximization problem on $\CP(\Gamma)$, we solve it on the
(open subset of) linear space $\CD^+(\Gamma) \subset \CD(\Gamma)$ by considering \eqref{eq:intrho=1} itself
as a constraint like \eqref{eq:pi=intFi}.

In particular at any extremum point $D \in \CP(\Gamma)$, there exist some constants  $(w,\lambda_1, \ldots, \lambda_n)$ such that
we can express the one-form $d\CS|_D$
as a linear combination
\begin{equation}\label{eq:conormal-condition}
d\CS|_D = w \cdot d M|_D - \sum_i \lambda_i d\CO_{F_i}|_D
\end{equation}
(``The method of the Lagrange multiplier'').
 Presence of the negative sign
is nothing significant but just a traditional convention.

More explicitly, by substituting $\dot D$ into \eqref{eq:conormal-condition}, we derive the identity
$$
\dot{D}[\CS] = w \dot D[M] -\sum_i  \lambda_i \dot D[\langle F_i \rangle]
$$
which is equivalent to the vanishing of the directional derivative
$$
\dot D\left[\CS - w  M + \sum _i \lambda_i \CO_{F_i} \right] = 0
$$
for all $\dot D \in \dot \CD^+(\Gamma) \cong T_D \CD^+(\Gamma)
\cong \CD(\Gamma)$. This
can be explicitly written into
$$
-\int_\Gamma \dot D \left(\log{\frac{\del D}{\del \nu_0}
+w- \sum_I \lambda_i F_i}\right) = 0
$$
by the linearity of $\CO_{F_i}$ and by definition of the mass
function $M$. By writing
$$
f_D := \frac{\del D}{\del \nu_0},
$$
we obtain
$$
\log f_D + w -\sum_{i=1}^n \lambda_i F_i= 0
$$
and hence
\begin{equation}\label{eq:fD}
f_D = e^{-w + \sum_i \lambda_i F_i}
\end{equation}
Therefore we have proved
\begin{prop} At any extreme point $D$ of the entropy under the constraint
\eqref{eq:pi=intFi}, there exists a Lagrange multiplier
$(w, \lambda_1, \ldots, \lambda_n)$ such that
$D$ has the form
$$
D = e^{-w+ \sum_i \lambda_i F_i}\, \nu_0
$$
\end{prop}

Then \cite{MNSS} proceeds further using the normalization condition
$$
1 = \int_\Gamma D = \int_\Gamma e^{-w+ \sum_i \lambda_i F_i}\, \nu_0
$$
which gives rise to the expression of $w$ in terms of other Lagrange multipliers
\begin{equation}\label{eq:ew}
e^w = \int_\Gamma e^{\sum_i \lambda_i F_i}\, \nu_0.
\end{equation}
By taking the logarithm thereof, we get
\begin{equation}\label{eq:w}
w = \log \left(\int_\Gamma e^{\sum_i p\lambda_i F_i}\, \nu_0\right).
\end{equation}
\begin{cor}\label{cor:MNSS}
For the probability density, $\rho_{\lambda}=e^{-w(\lambda)+\sum_i \lambda^iF_i}\, \nu_0$,
any equilibrium state of TPS associated to the observables $\CF = \{F_1, \ldots, F_n\}$ and the Lagrange multipliers $\lambda=(\lambda^1,\dots,\lambda^n)$ in the previous proposition,
the followings also hold:
\begin{enumerate}[$(1)$]
\item
At any extreme point $\rho$, $\rho$ has the form
\begin{equation}\label{eq:rho-formula}
\rho_\lambda =  \left(\frac{e^{\sum_i \lambda_i F_i}}{\int_\Gamma
e^{\sum_i \lambda_i F_i}} \right) \, \nu_0.
\end{equation}
\item The observations associated to the variables $\lambda = (\lambda_1, \cdots, \lambda_n)$
have their values
$$
q^i(\lambda):=\int_\Gamma F_i\rho_{\lambda} = \frac{\partial{w}}{\partial{\lambda_i}}
$$
\item We denote by $\rho_{(w,q,p)}: = \rho_\lambda$
the probability density $\rho_\lambda$ on $\Gamma$
determined above. Then in terms of the variables
$$
(\lambda_1, \ldots, \lambda_n) =: (p_1, \dots, p_n)
$$
the(relative) entropy is expressed as
$$
\CS (\rho_{(w,q,p)} )= - \int_{\Gamma} \rho_{(w,q,p)}
\log \frac{\del \rho_{(w, q,p)}}{\del \nu_0}.
$$
\end{enumerate}
\end{cor}

\begin{rem}
\begin{enumerate}
\item \cite{MNSS} goes further  and consider the expressions
\beastar
s & = & -\log \rho = w - \sum_i p_iF_i \\
ds & = & dw - \sum_i F_i dp_i
\eeastar
by saying ``\emph{These are both functions of the microscopic variables $\Gamma$
(by way of $F_i(\Gamma)$), and of the parameters $w$ and $p_1, \cdots, p_n$. Nevertheless,
differentiation is understood to be only with respect to the variables
$w, \, p_1, \cdots, p_n$}.''
\item The upshot of Section \ref{sec:reduced-entropy} is that it provides precise mathematical
statement  in the above quotation mark in a mathematically consistent way.
\end{enumerate}
\end{rem}

\section{Statistical description of gas in a piston and  volume variable}
\label{sec:volume}

So far we have considered the configuration space $\R^{3N}$, i.e.,
the set of $N$
\emph{identical} particles on the full space $\R^3$.
In this section, we explain in our framework how we can involve
the thermodynamic volume variable, which is \emph{not}
a local observable in the gas model of statistical mechanics.

Consider a thermodynamic system of gas, (i.e., a system of a finite fixed number of particles) within
the cylindrical container
$$
D^2(r) \times (0, h)
$$
that is (vertically) expandable from the base. In other words, we assume that its volume varies by some outside work.
  Following
the common words, we will call a \emph{piston} any expandable container.

We denote by $h > 0$ the parameter with respect to which
the volume \emph{grows linearly} in terms of the variable $h$.
We may regard $h$ the height of the cylinder at a given moment and so
 the volume of the container in $\R^3$ associated to the parameter $h$ will be given by
$$
F_2(h) = C h
$$
for some fixed constant $C > 0$.

Then we consider the configuration space $M \subset \R^K \times \R_+$ with $K = 3N$
defined by
\begin{equation}\label{eq:M}
M =  \bigcup_{h \in \R^+}  (D^2 \times (0,h))^N \times \{h\}
\end{equation}
which is an open subset of $\R^{3N} \times \R_+ \subset \R^{3N} \times \R = \R^{3N+1}$.
We write the natural coordinates of $M$ coming from $\R^{3N+1}$ as
$$
(q_1, \cdots, q_{3N}, h), \quad  h \in \R_+.
$$
We then consider
\begin{equation}\label{eq:T*MtoT^*R+}
\Gamma: = \bigcup_{\Lambda \in \R^+} T^*(D^2 \times (0,h))^{N} \times \{\Lambda\}
\end{equation}
and the natural projection $\pi_h: \Gamma \to \R_+$, where we denote
\begin{equation}\label{eq:Lambda}
 \Lambda: = \text{\rm vol}( (D^2(r) \times (0,h))^N)= (\pi r^2 h)^N.
\end{equation}
(Here we use $\Lambda$ instead of $h$ to emphasize that our SPS is $T^* \R^{3N}$, not $T^*\R^3$.)

We write the natural coordinates of
$$
\Gamma \cong \bigcup_{\Lambda \in \R_+}
(D^2 \times (0,h))^N \times \R^{3N} \times \{\Lambda\} \subset \R^{3N} \times \R^{3N} \times \R_+
$$
by
$$
(q_1, \cdots, q_{3N}, p_1, \cdots, p_{3N}, h).
$$
The canonical symplectic form on $ T^*(D^2 \times (0,h))^{3N}$
is given by the restriction of
$$
\Omega_N = \sum_{i=1}^{3N} dq_i \wedge dp_i.
$$
\begin{rem}
The SPS in the present case is a degenerate symplectic manifold,
or more precisely a \emph{presymplectic} manifold with its null
distribution is spanned by $\frac{\del}{\del h}$:
For our purpose of deriving the equations of the ideal gas or the van der Waals models, $h$ does not appear as a dynamical variable
and so we do not need to consider dynamics on the variable $h$.
An upshot of our framework of the gas in a piston is that
\emph{in the statistics or in the information theoretic point of view},
we regard the motion of the gas in the piston as the composition
of two independent `events', one followed by the other:
\begin{enumerate}
\item [{(E1)}]The \emph{configuration space motion} of the piston encoded by the parameter
$h \in \R_+$,
\item [{(E2)}] The \emph{phase space motion} of the gas inside the container.
\end{enumerate}
\end{rem}
We hope to come back to the study of a more general thermodynamical system
elsewhere, and restrict ourselves to the
cases of ideal gas and van der Waals models in the present paper.

This being said, let $\CF = \{F_1, \cdots, F_n\}$ on $T^*\R^{3N}$
be a system of local observables. We consider the observable system
on $\Gamma$ given by
$$
\{F_1, \cdots, F_n, \Lambda\} =: \CF*\Lambda
$$
where $F_i$ are local functions on $T^*C_h^{3N} \subset T^*\R^{3N}$, which are
independent of $\Lambda$, i.e., have the form
$$
F_i = F_i(q,p), \quad i = 1, \cdots, n.
$$
We apply all of our construction preformed in the previous sections to this observable system
with the following modifications.

We define the collective observation to be
\begin{equation}\label{eq:new-Jg}
J_{\CF* \Lambda} = (\CO_{F_1}, \cdots, \CO_{F_n},\Lambda)
\end{equation}
and write the observation data set
$$
B_{\CF*\Lambda} := \Image J_{\CF*\Lambda} \subset \mathfrak{g}_{\CF*\Lambda}
$$
Then we have natural fibration
$$
B_{\CF*\Lambda} \to (0,\Lambda)
$$
and hence the $\CF*\Lambda$-reduced KTPS carries the variable $\Lambda$
only as a configuration variable \emph{without} conjugate partner in the SPS level. (Recall the direction of $h$ is the null direction of the
symplectic manifold $\Gamma$.)
By taking the relevant SPS and KTPS to be
$$
\Gamma = T^*C_h^{3N}  \times \R_+
$$
and
$$
T^*\CP(\Gamma)
$$
respectively, we still take the associated $\TPS$ to be
$$
\CT = J^1 (\mathfrak{g}_{\CF*\Lambda}).
$$
Then, similarly to \eqref{eq:generating-diagram}, the $\CF*\Lambda$-reduced
entropy function $\CS^{\text{\rm red}}_{\CF*\Lambda}$ has the diagram
\begin{equation}\label{eq:generating-diagram3}
\xymatrix{\CM_{\CF*\Lambda} \ar[r]^{\CS^{\text{\rm red}}_{\CF*\Lambda}} \ar[d] &\R \\
B_{\CF*\Lambda} &
}
\end{equation}
and generates the thermodynamic equilibrium of $\CF*\Lambda$. In addition to
the variables $(p_1, \cdots, p_n)$ with $p_i = \langle F_i \rangle_\rho$,
we have additional variable coming from $\Lambda$. We call this the \emph{volume variable}
of the system $\CF*\Lambda$ and its conjugate variable in
$J^1(\mathfrak g_{\CF * \Lambda}) \cong
T^*( \mathfrak g_{\CF * \Lambda})  \times \R$
the \emph{pressure variable} and denote it by $P$.

\section{Thermodynamics of ideal gas}
\label{sec:ideal-gas}

In this section and the next, we recall the well-known study on the equilibrium gas equations
in the point of view of statistical mechanics and reconstruct them by  our geometric reduction.

Deriving the gas equation in an equilibrium state in thermodynamics means to
express an element in our Legendrian submanifold $R_{\CF;\CS}$ in the form
$$
(df(p),p, f(p))
$$
for some function $f: {\R}^n_{\CF} \to \R$ called a \emph{thermodynamic potential}.
However such a global function may not exist in general as later illustrated by
the van der Waals gas model.

\begin{rem} We would like to attract readers' attention to the way how we write the
expression $(df(p),p,f(p))$ differently from $(q,df(q),f(q))$ which is the standard way
of writing the 1-jet graph of $f$ in contact geometry. We will see
that this is reason why the \emph{horizontal} line
is drawn in the Maxwell construction, not the \emph{vertical} line drawn in the
construction of \emph{graph selectors} in the cotangent or in the 1-jet bundle in
symplectic and contact geometry. (See \cite{oh:jdg}, \cite{oh:C0-continuity}, \cite{amorim-oh-santos}.)
\end{rem}

\subsection{Nonequilibrium thermodynamics of ideal gas}

Consider the ideal gas of $N$ identical particles contained in the piston $C^{3N}_h$
described in Section \ref{sec:ideal-gas}.
Recall that in the ideal gas model we start with the kinetic energy Hamiltonian
$$
K = \frac1{2m}\sum_{i=1}^N |P_i|^2
$$
on $\SPS$ which drives the dynamics of many-body system not only on
$SPS$ but also induces the kinetic equation
\eqref{eq:drhodt} which is
equivalent to the equation for the density
$$
\frac{\del f}{\del t} = - \{f, K\}, \quad f=f_\rho(t,Q,P)
$$
The Poisson bracket of the right hand side is computed to be
$$
\{f, K\} = P \cdot \nabla_Q f
$$
where $\nabla_Q f$ is the $Q$-gradient
$$
\nabla_Q f = \left(\frac{\del f}{\del Q_1}, \cdots, \frac{\del f}{\del Q_K}\right)
$$
and hence the equation is reduced to
\begin{equation}\label{eq:dfdt-2}
\frac{\del f}{\del t}  + P \cdot \nabla_Q f = 0.
\end{equation}

In this subsection, we would like to describe the reduced flow of this equation on TPS, i.e.,
on $J^1\R^2_{(U,V)}$ in terms of the coordinates $(U,V,T,P,S)$ for
$$
U: = \langle K \rangle_\rho, \quad S: = \CS(\rho).
$$
For this purpose, let us first give concrete description of the quotient space
$$
\CM^\CF_{\text{\rm int}} : = \bigcup_{\mu \in B_\CF^\circ} \CM^\CF_\mu = \bigcup_{\mu \in B_\CF^\circ} \CJ_\CF^{-1}(\mu)/\CG_\CF:
$$
In the current situation, the system $\CF * V$ consists of the kinetic energy $K$ which is a \emph{local observable}, and the volume $V$ which is a \emph{nonlocal one.} Therefore, dropping the nonlocal variable $V$ from the discussion for the moment, we will
focus on the part of local observable and consider the reduced KTPS
\begin{equation}\label{eq:reduced-KTPS}
\CJ_{\{K\}}^{-1}(\mu)/\CG_\CF
\end{equation}
and the $\{K\}$-reduced entropy function $\CS_\CF^{\text{\rm red}}: \CJ_{\{K\}} ^{-1}(\mu)/\CG_\CF \to \R$.

We first note that each element in \eqref{eq:reduced-KTPS} represents a trajectory of
a point $(\rho,\beta) \in \CJ_{\{K\}}^{-1}(\mu) \subset T^*\CP(\Gamma)$ induced by the Hamiltonian flow of
$K$ on SPS $\Gamma$ with energy expectation value (or energy observation)
$\mu$. In other words it is the
moduli space of solutions of the system of PDE
\begin{equation}\label{eq:kinetic-eq}
\frac{\del \rho}{\del t} = - \CL_{X_K} \rho, \quad \frac{\del \beta}{\del t} = \CL_{X_K} \beta
\end{equation}
for $(\rho,\beta)$ in
$$
T^*\CP(\Gamma) \cong \CP(\Gamma) \times \CD_0(\Gamma):
$$
We recall the action of $\CG_{\{K\}}$ on $T^*\CP(\Gamma)$ is given by
$$
(\phi, (\rho,\beta)) \mapsto (\phi_*\rho, \phi^*\beta).
$$
We now consider the time-evolution of the Lagrangian submanifold
$$
\Image d\CS \subset T^*\CP(\Gamma)
$$
under the flow \eqref{eq:kinetic-eq}. By the relation, $\beta = d\CS(\rho)$
on $\Image d\CS \subset T^*\CP(\Gamma)$, the evolution thereof is
determined by the evolution of $\rho$-component which is
$$
\frac{\del \rho}{\del t} + \CL_{X_K} \rho = 0
$$
with the substitution of $P = m v_0$ for a constant velocity
This equation is equivalent to
\begin{equation}\label{eq:kinetic-eq2}
\frac{\del f}{\del t} + mv_0 \cdot \nabla_Q f = 0.
\end{equation}
Since $v_0 \cdot \nabla_Q$ is the vector field generating the translation $t \mapsto q_0 + v_0 t$,
we have the unique solution $f(t,q) = f_0(q_0 - v_0t)$ with the initial probability condition 
$\rho_0 = f_0\, \nu_0$ \emph{assuming suitable regularity of $f_0$}.

\emph{One may interpret this kinetic equation as the mesoscopic nonequilibrium 
thermodynamics of the ideal gas.}

\subsection{Ideal gas equation at equilibrium}

In general thermodynamic equilibrium, say in the system $\CF$ of ideal gas or van der Waals gas equation, 
the state variables are given by
$(U, T, V, P, S)$:
\begin{itemize}
\item $U = \langle H \rangle_\rho$ is the energy,
\item $S := \CS(\rho)$ is the thermodynamic entropy,
\item $T$ is the temperature,
\item $V$ is the volume, and $P$ the pressure.
\end{itemize}
The first law of equilibrium thermodynamic is written as
$$
dS = \frac1T dU + \frac{P}{T} dV
$$
in the entropy representation.

For the ideal gas,
the energy and the volume of $(D^2(r) \times (0,h))^{3N}$
are given as
\beastar
F_1 &=& \frac{1}{2m}\sum_i^{N} |P_i|^2 \\
F_2 &=& C h, \quad C = \pi r^2
\eeastar
for some positive constant $m$. Set
$P = (P_1, \cdots, P_{3N}) \in \R^{3N}_p$
and
$$
Q = (Q_1, \cdots, Q_{3N},h) \in M
$$
for $M$ given in \eqref{eq:M}.
We have the formulae for the variable $U$ and $V$
\begin{equation}\label{eq:UV}
U = \langle F_1 \rangle_\rho, \quad V = \langle F_2 \rangle_\rho
\end{equation}
by definition. 

We now provide the explicit description
of the equilibrium of the ideal gas by deriving its gas equation.
We will express the thermodynamic entropy $S$ as a
function of $(U,V)$ by computing the formulae for the expectation value
of $F_1$ and $F_2$ on
$$
\Gamma = T^*\R^{3N} \times \R_+
$$
with respect to the probability
distribution $\rho$ thereon.

\subsubsection{Calculation of $U = \langle F_1 \rangle_{\rho_\lambda}$}

For the probability distribution $\rho$ on $T^*\R^{3N}$
or on its open subsets,
the procedure given in Section \ref{sec:review} provides $\rho_\lambda = f \, d\Gamma$ with
\begin{equation}\label{eq:frho2}
f = \frac{ e^{\frac{\lambda_1}{2m} \sum_i |P_i|^2 + \lambda_2 Ch}}{\int_{Q^{3N}_h} e^{\frac{\lambda_1}{2m}
\sum_i |P_i|^2 + \lambda_2 Ch }\, d\Gamma}.
\end{equation}
Therefore
\begin{equation}
\langle F_i \rangle_{\rho_\lambda} = \frac{
 \int F_i e^{ \frac{\lambda_1}{2m} \sum_i |P_i|^2 + \lambda_2 Ch} \, d\Gamma}
{\int_{Q^{3N}_h} e^{\frac{\lambda_1}{2m} \sum_i |P_i|^2 + \lambda_2 Ch\, d\Gamma}  \, d\Gamma}.
 \end{equation}
We compute the expectation value
\beastar
\langle F_1 \rangle_{\rho_\lambda}
& = & \frac{\int_{\Gamma} \left(\frac1{2m} \sum_{i=1}^N |P_i|^2\right)
e^{\lambda_1 \left(\frac1{2m} \sum_{i=1}^N |P_i|^2\right) + \lambda_2 C h} dP dQ dh}
{\int_{\Gamma}
e^{\lambda_1 \left(\frac1{2m} \sum_{i=1}^N |P_i|^2\right) + \lambda_2 C h} dP dQ dh}\\
& = & \frac{\int_{\R^{3N}} \left(\frac1{2m} \sum_{i=1}^N |P_i|^2\right)
e^{\lambda_1 \left(\frac1{2m} \sum_{i=1}^N |P_i|^2\right) } dP}
{\int_{\R^{3N}}
e^{\lambda_1 \left(\frac1{2m} \sum_{i=1}^N |P_i|^2\right) } dP}
\eeastar
where $dP = \wedge_{i=1}^N (dp_i^1 \wedge dp_i^2 \wedge dp_i^3)$
with $P_i = (p_i^1,p_i^2,p_i^3)$. We evaluate the integral using the following recurrence
relation. We set
$$
E_N: =  \frac{\int_{\R^{3N}} \left(\frac1{2m} \sum_{i=1}^N |P_i|^2\right)
e^{\lambda_1 \left(\frac1{2m} \sum_{i=1}^N |P_i|^2\right) } dP}
{\int_{\R^{3N}}
e^{\lambda_1 \left(\frac1{2m} \sum_{i=1}^N |P_i|^2\right) } dP}
$$
and $U_N = \frac1{2m} \sum_{i=1}^N |P_i|^2$.

\begin{lem} For $N \geq 2$, we have
$$
E_N = E_1 + E_{N-1},
$$
and $E_1 = -\frac{3}{2 \lambda_1}$.  In particular, we have
$E_N = -\frac{3N }{2 \lambda_1}$.
\end{lem}
\begin{proof} We can express the last integral as
$$
E_N = \frac{\int_{\R^3} \int_{\R^{3(N-1)}}
(U_1 + U_{N-1}') e^{\lambda_1 (U_1 + U_{N-1}')}
dP_1 \wedge dP_2 \wedge \cdots \wedge dP_N}{\int_{\R^3} e^{\lambda_1 U_1 \,
dP_1 \cdot
\int_{\R^{3(N-1)} }e^{\lambda_1 U_{N-1}'}dP_2 \wedge \cdots \wedge dP_N} }
$$
for $U_1 = \frac1{2m} |U_1|^2$ and
$U_{N-1}' = \frac1{2m}\sum_{i=2}^2 |U_i|^2$. It becomes
\beastar
& {}&  \frac{\int_{\R^3}U_1 e^{\lambda_1U_1}\, dP_1} {\int_{\R^3} e^{\lambda_1 U_1} dP_1}
+ \frac{\int_{\R^{3(N-1)}}e^{\lambda_1 U'_{N-1}} \, dP_2 \wedge \cdots \wedge dP_N)}{\int_{\R^{3(N-1)}}
e^{\lambda_1 U'_{N-1}} dP_2 \wedge \cdots \wedge dP_N} \\
& = & E_1 + E_{N-1}.
\eeastar
We now express
\beastar
E_1 & =  & \frac{\int_{\R^3}U_1 e^{\lambda_1U_1}\, dP_1} {\int_{\R^3} e^{\lambda_1 U_1} dP_1}
=\frac{\int_0^{2\pi }\int_0^\pi \int_0^\infty \frac{r^2}{2m} e^{\frac{\lambda_1}{2m} r^2} \,
 r^2\,\sin \theta  dr \, d\varphi\, d\theta  }
 {\int_0^{2\pi }\int_0^\pi \int_0^\infty e^{\frac{\lambda_1}{2m} r^2} \,
 r^2\,\sin \varphi  dr  \, d\varphi\, d\theta } \\
 & = &
 \frac{\int_0^\infty r^4 e^{\frac{\lambda_1}{2m} r^2} dr}
 {2m \int_0^\infty r^2 e^{\frac{\lambda_1}{2m} r^2} dr} = - \frac{3}{2\lambda_1}.
 \eeastar
\end{proof}
We note that for the integral to converge, we need to choose $\lambda_1 < 0$.
In that case  a direct calculation gives rise to
\begin{lem}\label{eq:U=}
\begin{equation}\label{eq:U}
U = \langle F_1 \rangle_{\rho_\lambda}  = - \frac{3N }{2 \lambda_1}.
\end{equation}
\end{lem}

\subsubsection{Calculation of $V$}

Similar calculation, which is easier, leads to
\beastar
\langle F_2 \rangle_{\rho_\lambda}
& = & \frac{\int_{\Gamma} (C h)
e^{\lambda_1 F_1 + \lambda_2 Ch} dQ \wedge dh}
{\int_{\Gamma}
e^{\lambda_1 F_1 + \lambda_2 Ch }  dQ \wedge dh} \\
& = & \frac{\int_{C^N_h \times \R_+} (C h) e^{\lambda_2 Ch}\, dQ \wedge dh}
{\int_{C^N_h \times \R_+} e^{\lambda_2 Ch}  \, dQ \wedge dh}
=  \frac{ \int_0^\infty Ch (\pi h)^N e^{\lambda_2 Ch} \, dh}
{ \int_0^\infty (\pi h)^N e^{\lambda_2 Ch} \, dh}
\eeastar
We now evaluate integral
\bea
\int_0^\infty (C h) (\pi h)^N e^{\lambda_2 Ch}   dh & = &
\left(\frac{-1}{\lambda_2 C}\right)^{N+2} (N+1)! C \pi^N
\nonumber \\
\int_0^\infty (\pi h)^N  e^{\lambda_2 Ch} dh & = &
\left(\frac{-1}{\lambda_2 C}\right)^{N+1} (N)! \pi^N.
\eea
This proves
\begin{lem}\label{eq:V=}
\begin{equation}\label{eq:V}
V  : = \langle F_2 \rangle_{\rho(w,q,p)} \,= - \frac{N+1}{\lambda_2}.
\end{equation}
\end{lem}

\subsubsection{Evaluation of $S$}

Finally we compute
\beastar
\CS(\rho_\lambda) & = & \int_\Gamma - \rho_\lambda \frac{\del \rho_\lambda}{\del \nu_0}
=\int_\Gamma (w -  \lambda_1 F_1 - \lambda_2 F_2) \rho_\lambda \\
& =  & w - \lambda_1 \langle F_1 \rangle_{\rho_\lambda} - \lambda_2
\langle F_2 \rangle_{\rho_\lambda}.
\eeastar
From \eqref{eq:frho2}, we have
$$
\rho_\lambda = f\, d\Gamma \frac{e^{\lambda_1 F_1 + \lambda_2 F_2} \, d\Gamma }{\int_\Gamma e^{\lambda_1 F_1 + \lambda_2 F_2} \, d\Gamma}.
$$
We also have
$$
w = \log \left(\int_\Gamma e^{\lambda_1 F_1 + \lambda_2 F_2} \, d\Gamma \right)
$$
from \eqref{eq:w} for which we write
$$
\int_\Gamma e^{\lambda_1 F_1 + \lambda_2 F_2} \, d\Gamma
= \int_{\R^{3N}} e^{\frac{\lambda_1}{2} |P|^2 \, dP} \cdot
\int_0^\infty (\pi h)^N e^{\lambda_2 Ch} \, dh.
$$
We have already computed
$$
 \int_0^\infty (\pi h)^N e^{\lambda_2 Ch}\, dh = \left(\frac{-1}{\lambda_2C}\right)^{N+1}
N! \pi^N
$$
and compute
$$
 \int_{\R^{3N}} e^{\frac{\lambda_1}{2} |P|^2} \, dP
 =  \left(2 \sqrt{-\frac{2\pi m}{\lambda_1}}\right)^N.
 $$
 This proves
$$
 w = \log \left( \left(\frac{-1}{\lambda_2C}\right)^{N+1}
N! \pi^N \right) + \log \left(2 \sqrt{-\frac{2\pi m}{\lambda_1}}\right)^N.
$$
Combining the above, we have computed
\begin{lem}\label{eq:w=}
\begin{equation}\label{eq:S}
w = \frac{3N}{2} \log\left(- \frac{2\pi m}{\lambda_1}\right)
+ (N+1)\log\left(-\frac{1}{\lambda_2 C}\right) + \log(N!) + N\log(2\pi).
\end{equation}
\end{lem}

\subsubsection{Derivation of ideal gas law}

At the equilibrium state $R_{\CF;\CS}$, a thermodynamic potential is
nothing but $\CS_\CF^{\text{\rm red}}$ with respect to the 
preferred variable $(U,V)$ restricted to the associated
Legendrian submanifold regarded as a global function in that a point
$(q,p,w) \in R_{\CF;\CS}$ can be expressed as
$$
q = \mu, \, p = df(\mu), \, w = f(\mu).
$$
In general such a function may or may not exist as a \emph{single-valued} function.
The following definition deserves attention which is a standard terminology.

\begin{defn}\label{defn:holonomic} Let $(q_1, \cdots, q_n, p_1, \cdots, p_n,w)$ be the
canonical  coordinates of the 1-jet bundle $J^1M$ and consider the canonical contact form
$\lambda = dw - \sum_{i=1}^n p_i dq_i$. We call a Legendrian submanifold $R \subset J^1M$
\emph{holonomic with respect to the  variables $(q_1,\cdots, q_n)$}
if it admits a function $g: M \to \R$ such that
$R = \Image j^1g$. We say $R$ is \emph{non-holonomic
with respect to the variables $(q_1,\cdots, q_n)$ otherwise}.
\end{defn}

We will see that for the ideal gas, such a function indeed exists which we denote by
$S = S(U,V)$  of the variables $(U,V)$. By definition, we have the formula
$$
S = - \int_\Gamma \rho_\lambda \frac{\del \rho_\lambda}{\del \nu_0}.
$$
\begin{prop} Let $R_{\CF;\CS} \subset \R^5$ be as above.
Then $R_{\CF;\CS} \subset \R^5$ is holonomic with respect to $(U,V)$ and
 have the formula for the thermodynamic potential
$$
S(U,V) = \frac{3N}{2}\log\left(\frac{4\pi m U}{3N}\right) +
(N+1)\log\left(\frac{V}{C}\right)  + \log(N!) + N \log 2\pi.
$$
\end{prop}
\begin{proof} From the explicit formulae for $\U$ and $V$ in \eqref{eq:U}, \eqref{eq:V},
we can \emph{invert} the  map $(\lambda_1, \lambda_2) \mapsto (U,V)$ and get
$$
\lambda_1 = - \frac{3N}{2U}, \quad \lambda_2 = -\frac{N+1}{V}
$$
into \eqref{eq:S}. The formula for $S$ \emph{as a single-valued} function follows.
\end{proof}

Since the conjugate variables of $U,\, V$ are
 $\frac{1}{T}, \, \frac{P}{T}$, we have derived the equation
\beastar
 \frac1{T}  & =  & \frac{\del S}{\del U} = \frac{3N}{2U} \\
 \frac{P}{T} & = &  \frac{\del S}{\del V} = \frac{N+1}{V}
 \eeastar
which is equivalent to the well-known ideal gas equation
\begin{cor}[Ideal gas equation] Let $N$ be the number of particles
in a piston of ideal gas. Then we have the relations
\begin{equation}\label{eq:ideal-gas-eq}
U = \frac{3}{2}NT, \qquad
V = (N+1) \frac{T}{P}
\end{equation}
(modulo a multiplication by constant).
\end{cor}
Here we would like to highlight the fact that the equilibrium
state $R_{\CF;\CS}$ has constant energy in the isotherm, and
hence it has the form of a cylinder over the one dimensional
curve in the $PV$-plane. In particular \emph{its
 $PV$-diagram undergoes no phase transition}.

\section{Phase transition, Maxwell construction and graph selector}
\label{sec:maxwell}

In our construction, a relevant thermodynamic equilibrium is given by
the Legendrian submanifold $R$ generated by the $\CF$-reduced  entropy function
$$
\CS^{\text{\rm red}}_\CF: \CM^{\CF} \to \R
$$
where $\CM^{\CF} \to \R^2 = \R^2_{U,V}$ is a (infinite dimensional) symplectic fibration.

If we change the preferred variables $(U,V)$ to something else, 
the coordinate representation of the resulting Legendrian submanifold is not necessarily holonomic
(or graph-like) in the given preferred variables. This  indeed
occurs for the case of the van der Waals gas for the preferred variables $(P,T)$.
As a result, the aforementioned thermodynamic variable cannot be globally holonomic
(at least so over the $PT$-plane or over $PU$-plane). An effort of correcting this deficiency
is precisely the celebrated Maxwell construction in equilibrium thermodynamics
the explanation of which is now in order.

\subsection{Van der Waals model and its phase transition}

The van der Waals gas is the model describing the gas whose assumptions are
more rigid than the ideal gas \cite{Moran:Thermo}.
It needs two more restrictions dictated by 2 additional parameters $a, \, b$:
\begin{enumerate}
    \item {\bf Potential hypothesis} : There exists attractive potential between
    the particles with the constant $a$.
    \item {\bf Volume hypothesis} : Each particle has non-zero size with
    the internal volume of  $b$.
\end{enumerate}
Then the observables, volume and energy, should have the form of
\beastar
U(\vec{x},\vec{p},h) &=& \frac{1}{2m}\sum_i |p_i|^2 - \frac{aN^2}{C}\cdot\frac{1}{h}\\
V(\vec{x},\vec{p},h) &=& Ch-Nb
\eeastar
for some constant $m$ and $C$.

Consider the new observables $X,Y$ defined as
\beastar
X &=& V + Nb = Ch \qquad \hskip1in \text{(volume of container)}\\
Y &=& U + aN^2\cdot \frac{1}{V+Nb} = \frac{1}{2m}
\sum_i |p_i|^2. \qquad \text{(kinetic energy)}
\eeastar
Note that the conjugate variables $(\lambda_X,\lambda_Y)$
of $(X,Y)$ are related to those $(\lambda_U,\lambda_V)$
conjugate to $(U,V)$ by the equation
$$
dS - \lambda_X dX -\lambda_Y dY = dS - \lambda_U dU - \lambda_V dV
$$
where the integration of $(Ch-bN)^N$ is the reason why the volume of the space each particles available is the only $(Ch-bN)$ for fixed $h$.
Therefore, we derive
\beastar
\lambda_X &=& \lambda_V+aN^2\cdot\lambda_U \frac{1}{(V+Nb)^2} \\
\lambda_Y &=& \lambda_U.
\eeastar
By the same kind of calculation as in the ideal gas case, we compute
\beastar
X &=& \frac{\partial{w}}{\partial{\lambda_X}}
= -(N+1) \frac{1}{\lambda_X} + bN
\simeq NT \frac{1}{ P +a \frac{N^2}{(V+Nb)^2}} + bN \\
Y &=& \frac{\partial{w}}{\partial{\lambda_Y}}
= -\frac{3}{2}(N+1) \frac{1}{\lambda_Y} \simeq \frac{3}{2} NT
\eeastar
after some truncation of higher order terms.
Therefore, we derive the Van der Waals gas equation
\begin{equation}\label{eq:van-der-Waals-eq}
\begin{cases}
(P +a \frac{N^2}{(V+Nb)^2})V = NT \\
U = \frac{3}{2} NT-a \frac{N^2}{(V+Nb)^2}.
\end{cases}
\end{equation}
In terms of the effective volume of container as
$$
V_{\text{\rm eff}} = Ch=V+Nb,
$$
the equation is equivalent to
\begin{equation}\label{eq:van-der-Waals-eq2}
\begin{cases}
(P+a \frac{N^2}{\widetilde V^2})(V_c-bN) = NT \\
U = \frac{3}{2} NT-a \frac{N^2}{\widetilde V^2}
\end{cases}
\end{equation}
where $\widetilde V: = V_{\text{\rm eff}}$.

We now identify the associated thermodynamic potential  relative
to the preferred variable $(T,P)$
by rewriting the  first thermodynamic law $dU = TdS - PdV$
 into the following form
\begin{equation}\label{eq:TP-potential}
d(U+PV - TS) =-SdT + VdP.
\end{equation}
The function $\Upsilon : \R^5 \to \R$ defined by
\begin{equation}\label{eq:Upsilon}
\Upsilon = U + PV -TS (= H - TS)
\end{equation}
is nothing but the well-known Gibbs free energy (or \emph{available energy} in Gibbs' own term) \cite{gibbs},
where $H: = U + PV$ is call the \emph{enthalpy} in thermodynamics.
This $\Upsilon$ will play the role for the preferred thermodynamic variables
$(P,T)$.  (We avoid using the common alphabet `$G$' to write the Gibbs free energy
not to confuse readers here since the same alphabet $G$ appears in
the graph of Figure \ref{fig:graph2}
where we use the same alphabets as Maxwell did in \cite{maxwell:construction}.)

\begin{ques} Can we express the equilibrium Legendrian submanifold $R_{\CF;\CS}$
as the image of one-jet graph of differentiable function $f: \R^2_{(T,P)} \to \R$, i.e.,
such that
$$
R_{\CF;\CS} = \left\{ (q,df(q), f(q)) \, \Big|\, q \in \R_{(T,P)}^2 \right\}?
$$
I.e., is $R_{\CF;\CS}$ is holonomic with respect to the variables $(T,P)$?
\end{ques}

We will now illustrate by the van der Waals model that such a
task is not possible in general, and that the
\emph{Maxwell construction} is precisely
the outcome of Maxwell's effort to overcome this deficiency.
It is easy to show that the change of coordinates
$$
(U,T,P,V,S) \mapsto (T,-S, V, P, \Upsilon)
$$
preserves the thermodynamic contact structure
$$
\xi: = \ker\left(dS - \frac1T dU + \frac{P}{T} dV\right) ( = \ker (dU - TdS + PdV)).
$$
With this preparation, we are now ready to do the analysis of van der Waals model.
When $T = T_0$ is given, the first equation of
\eqref{eq:van-der-Waals-eq} becomes
a cubic equation of $V$,
$$
\left(P+a \frac{N^2}{\widetilde V^2}\right)(\widetilde V-bN)
= N T_0.
$$
When $P = P_0$ is given in addition to $T$, we have the cubic equation of $\widetilde V$
\begin{equation}\label{eq:Vc-eq}
\widetilde V^3 + \alpha \widetilde V^2 + \beta \widetilde V + \gamma = 0
\end{equation}
with
$$
\alpha = -\left(bN + N \frac{T_0}{P_0}\right), \, \beta = \frac{aN^2}{P_0}, \,
\gamma = - \frac{abN^3}{P_0}.
$$
The second equation of \eqref{eq:van-der-Waals-eq2} is of the form
\begin{equation}\label{eq:U2}
U = - \frac{aN^2} {\widetilde V^2} + \frac32 NT_0.
\end{equation}
\begin{rem} A fundamental difference of the van der Waals equation \eqref{eq:U2} from
the ideal gas equation \eqref{eq:ideal-gas-eq} is that the energy
$U$ is no longer constant in the \emph{isotherm} of the $PV$-diagram.
\end{rem}

The cubic equation \eqref{eq:Vc-eq} shows that there is a phase transition in the
$PV$-diagram at
the critical temperature $T = T_c$ at which the discriminant
\begin{equation}\label{eq:discriminant}
D = 18\alpha \beta \gamma- 4\alpha^3 \gamma + \alpha^2\beta^2 - 4\beta^3 - 27\gamma^2
\end{equation}
of the cubic polynomial vanishes: When $T < T_c$, \eqref{eq:Vc-eq} has 3 distinct roots, while when $T > T_c$,
it has a unique root.

\subsection{Metamorphosis of the isotherms in the $PV$-diagram}
\label{subsec:PVdiagram}

In thermodynamic equilibrium, a necessary condition
for stability is that pressure
$P$ does not increase with volume $V$, i.e., that \emph{volume should decrease as
pressure increases}. As we see in Figure \ref{fig:graph1}, this fails when temperature
$T$ is below the critical temperature $T_c$. The Maxwell construction  is a
way of correcting this deficiency by  considering a \emph{continuous} thermodynamic potential function, \emph{if we can}.

\begin{figure}[htbp]
\centering
\def\svgwidth{10pt}
\includegraphics[scale=0.5]{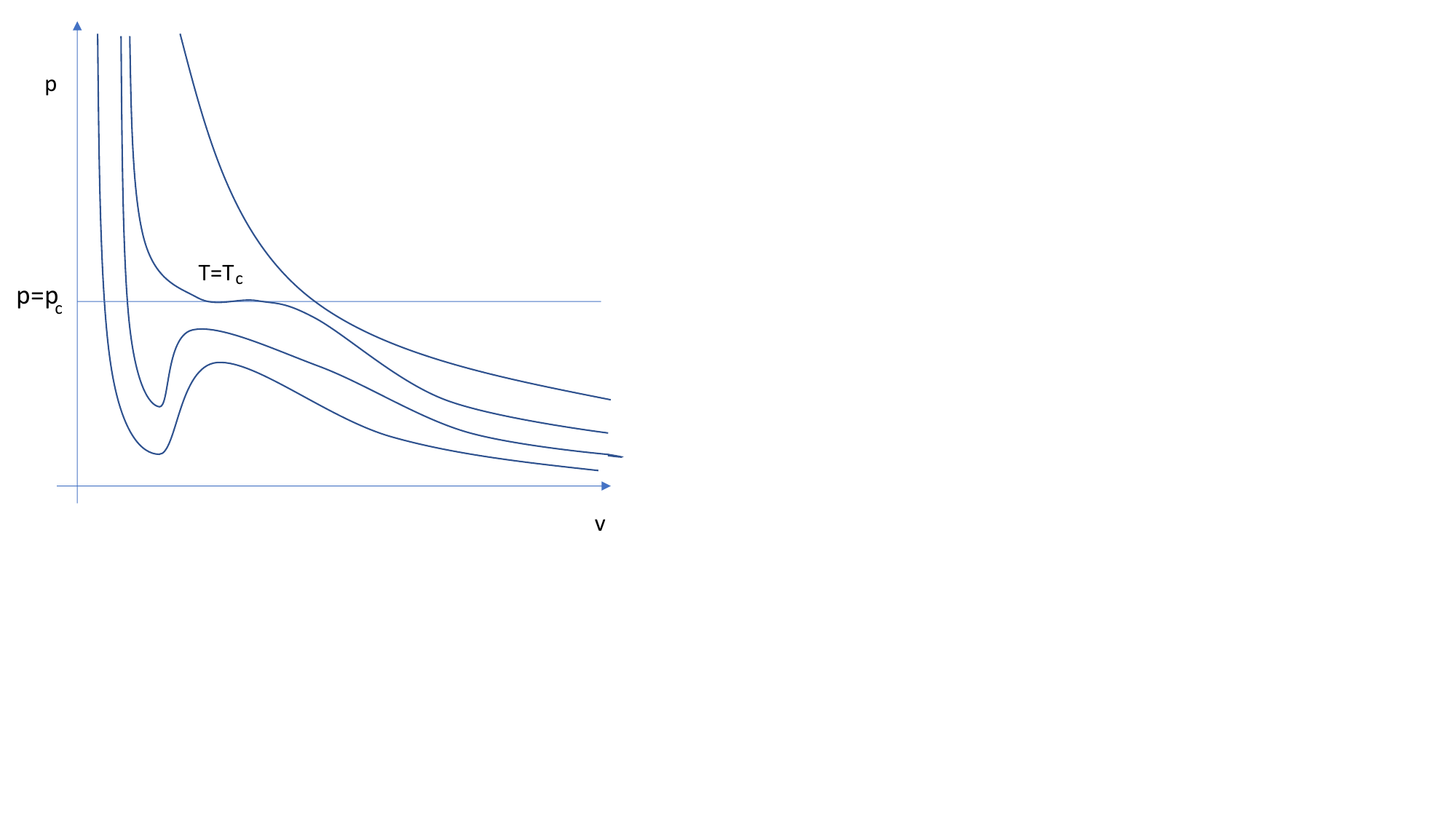}
\caption{Phase transition of $PV$-diagram of isotherm}
\label{fig:graph1}
\end{figure}

The above two cases can be also differentiated by the geometry of associated
equilibria whose explanation is now in order.

For this purpose, we first introduce the notion of
\emph{graph selector} of
a Legendrian submanifold $R \subset J^1M$.
The following definition is the Legendrian analog in  the 1-jet bundle $J^1M$
to the notion of graph selectors for the Lagrangian submanifolds
in the cotangent bundle $T^*M$.
(See \cite{arnaud}, \cite{amorim-oh-santos} for the definition of symplectic  graph selectors.)

\begin{defn}[Legendrian graph selector] \label{defn:graph-selector}
Let $\iota: N \to J^1M$ be a compact, Legendrian embedding
of $J^1M$. A \emph{graph selector} of a Legendrian submanifold $R \subset J^1M$
is a Lipschitz function $f:M\to\mathbb{R}$ such that $f$ is differentiable
on a dense open set $\mathcal{U} \subset M$ of full measure
and for all points $q \in \mathcal{U}$ we have
$$
(q,df(q), f(q))\in R.
$$	
\end{defn}

With this definition in our disposal, we can differentiate the two cases as follows.

\begin{thm}\label{thm:maxwell-construction}
Let $T_c$ be the aforementioned critical temperature. Then
\begin{enumerate}
\item When $T > T_c$, $R_{\CF;\CS}$ is a 1-jet graph of a function
$f: \R^2_{P,T} \to \R$.
\item When $T < T_c$,
\begin{enumerate}
\item such a globally single-valued function does not exist.
\item  Maxwell construction provides a graph selector
in the sense of  Definition \ref{defn:graph-selector} over the $PT$-plane, i.e., regarding
$(P,T)$ as the variables.
\end{enumerate}
\end{enumerate}
\end{thm}
\begin{proof} Statement (1) is obvious and so we focus on Statement (2).

In the latter case, we have $0 < T < T_c$.
From the first equation of \eqref{eq:van-der-Waals-eq}
$$
\left(P+a \frac{N^2}{\widetilde V^2}\right)(\widetilde V-bN) = k_BN T_0,
$$
we obtain the pressure formula
\begin{equation}\label{eq:P}
P = - a \frac{N^2}{\widetilde V^2} + \frac{k_B NT_0}{\widetilde V - bN}.
\end{equation}
There exists a unique critical pressure $P_c = P_c(T_c)$ such that
\begin{itemize}
\item when $P > P_c$, \eqref{eq:Vc-eq} has a unique solution $\widetilde V$,
\item when $P< P_c$, it has 3 distinct values of $\widetilde V$.
\end{itemize}
(See Figure \ref{fig:graph1}.)

For each $0 < T < T_c$, there is a unique pressure
denoted by $P_{\text{\rm mx}} = P_{\text{\rm mx}}(T) $ which we call 
\emph{Maxwell pressure} at which \emph{Maxwell equal area law holds}. (See Figure \ref{fig:graph2}.)

\begin{figure}[htbp]
 \centering
 \def\svgwidth{10pt}
\includegraphics[scale=0.5]{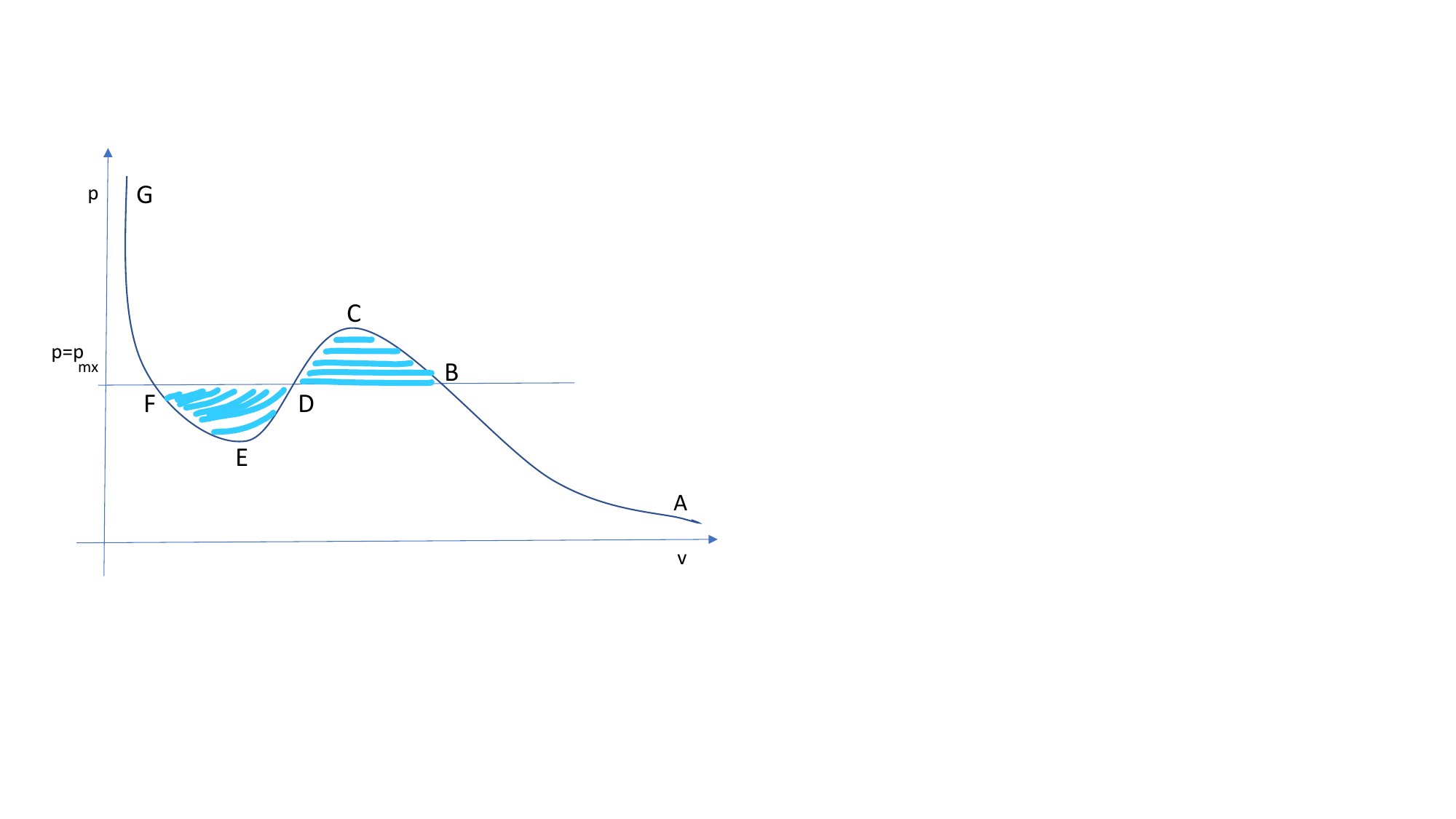}
\caption{Maxwell equal area law}
\label{fig:graph2}
\end{figure}

Now we will attempt to find a function of the type $f = f(P,T)$ so that
\begin{equation}\label{eq:graph-selector}
\left(T,P, \frac{\del f}{\del T}(T,P), \frac{\del f}{\del P}(T,P), f(T,P)\right) \in R_{\CF;\CS}
\end{equation}
explicitly by integration regarding the $PV$-diagram as the differential of a multi-valued
function whose differential is determined by the Legendrian submanifold $R_{\CF;\CS}$ (See \cite{arnaud,amorim-oh-santos} for such a procedure.):
Regard $\R^2_{(P,V)}$ as the cotangent bundle $T^*\R_P$ and
then the isotherm in the $PV$-diagram can be regarded as the Lagrangian submanifold
$$
(P,V) \in  \R^2_{(P,V)} \cong T^*\R_P
$$
where $\alpha = V dP \in T_P^*\R_P$.

To see how the aforementioned thermodynamic stability fails,
we consider the commutative diagram of the projections
$$
\xymatrix{\R_{(V,P)}^2 \ar[d] & \ar[l] R_{\CF;\CS} \cap \{T= T_0\} \ar[r] \ar[d]  & R_{\CF;\CS}
\ar[d] \\
\R_P & \ar[l] \R_{(P,T_0)} \ar[r] & \R_{(P,T)}^2
}
$$
and try to find a \emph{single-valued} section
$\Upsilon = \Upsilon(P,T)$ of the RHS projection. The $PV$-diagram
is the union of the projections of the top middle term to
the top left plane  $\R^2_{(V,P)}$ as $T_0$ varies.

\subsection{Maxwell construction and non-differentiable thermodynamic potential}
\label{subsec:maxwell}

\emph{For the simplicity of notation
for the following discussion, we will just write $V$ for $\widetilde V$ with $V$ 
representing the effective volume of container.}

We then consider the \emph{front projection}
$$
R_{\CF;\CS} \to \R_{(T,P,S)}^3; \quad (U,T,V,P,S) \to (T,P,S)
$$
and followed by the projection $\R_{(T,P,S)}^3 \mapsto \R_{(T,P)}^2$.
We call the latter composition $R_{\CF;\CS} \to \R_{(T,P)}^2$ the \emph{base projection}.
At each fixed temperature $T$, we would like to express the base projection of the isotherm
of the equilibrium $R_{\CF;\CS}$ as the graph $(P, g_T(P))$
of a single-valued function $V = g_T(P)$, \emph{if possible}.

Recalling that $P$ and $V$ are conjugate variable, we make the
 identification
$$
\R_{(P,V)}^2 \cong T^*\R_P
$$
so that the $V$-coordinate is the one for the coordinate representation
of the $\alpha \in T_P^* \R_P$ as
$$
\alpha = V \, dP, \quad V = \frac{\del f_T}{\del P}(P)
$$
for a germ $g$ of functions at $P$. As seen
in the isotherm pictured around the shaded region
 in Figure \ref{fig:graph2}, it has 3 different  branches of the
 graph seen horizontally.

Here we attempt to construct the function $f_T$
by assigning the value $f_T(P)$ of $f_T$ as follows:
Denote by $C^{\text{\rm min}}$ and $C^{\max}$
be the minimal and the maximal branches of the isotherm
in the $PV$ diagram for the region $P > P_{\text{\rm mx}}(T)$
and $P < P_{\text{\rm mx}}(T)$ respectively. We express
\beastar
C^{\text{\rm min}} & = & \{(P, \alpha^-(P)) \mid P > P_{\text{\rm mx}} \} \\
C^{\text{\rm max}} & = & \{(P, \alpha^+(P)) \mid P < P_{\text{\rm mx}} \}.
\eeastar
Then we define
$$
f_T(P) = \begin{cases} - \int_P^\infty \alpha^-(p) \, dp \quad & \text{\rm for }\, P > P_{\text{\rm mx}} \\
- \int_{P_{\text{\rm mx}}}^\infty \alpha^-(p) \, dp  - \int_{P_{\text{\rm mx}}}^P \alpha^+(p) \, dp
\quad & \text{\rm for }\, P < P_{\text{\rm mx}}.
\end{cases}
$$
It follows \emph{from Maxwell's equal area condition}
 that this function $g_T$
\emph{continuously} extends across the Maxwell pressure point
$P = P_{\text{\rm mx}}(T)$.
(See Figure \ref{fig:graph3}.)

\begin{figure}[htbp]
\centering
\def\svgwidth{10pt}
\includegraphics[scale=0.5]{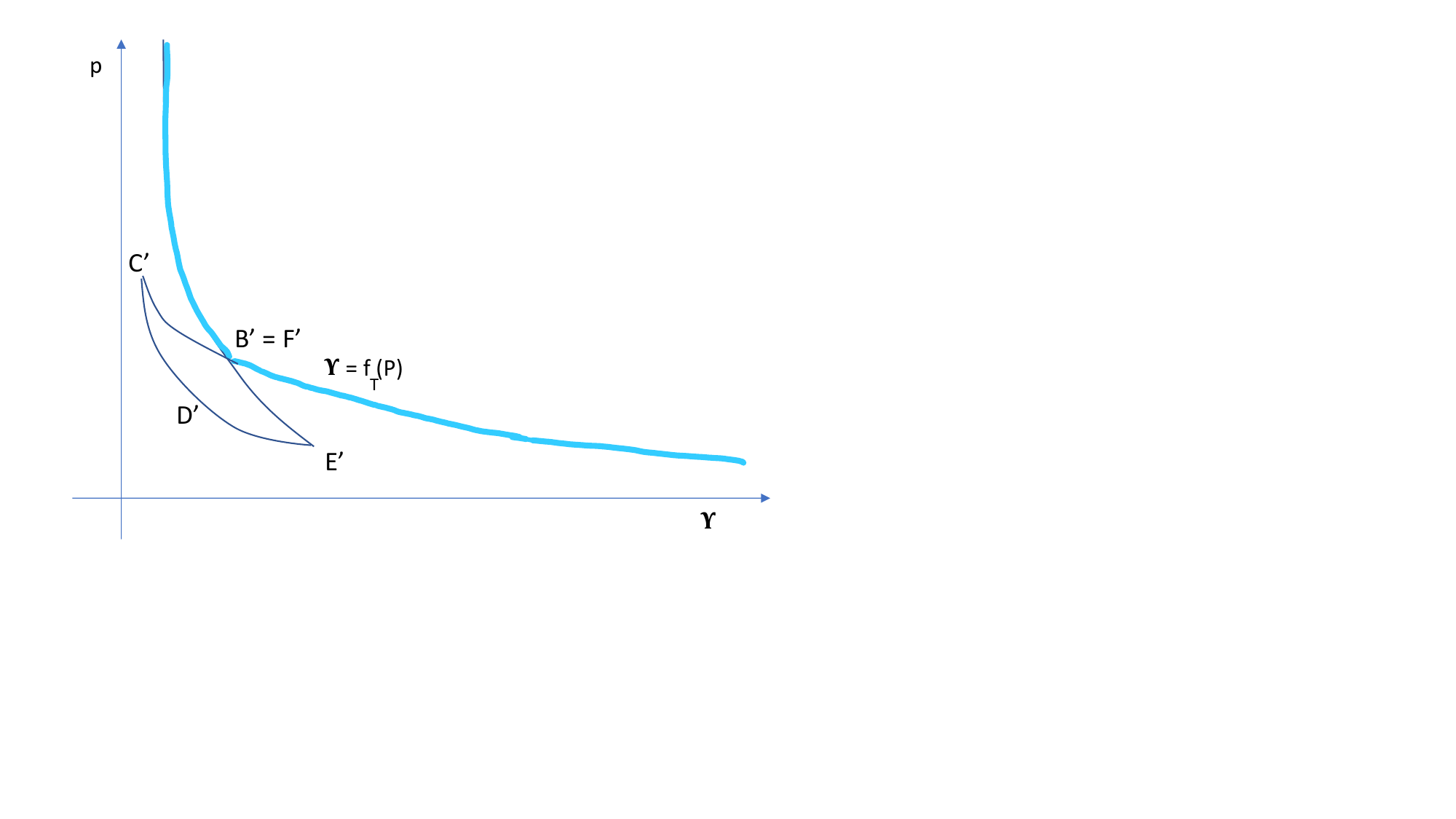}
\caption{Maxwell construction 1}
\label{fig:graph3}
\end{figure}

By varying $T$, this determines $V$ as a continuous function of
$(T,P) \mapsto f_T(P)$.  Then  substituting we determine $U = U(P,T)$ from
\eqref{eq:van-der-Waals-eq}.

Next we determine  the thermodynamic potential $S = S(P,T)$
by integrating the equation of the thermodynamic first law
$$
dS = \frac{1}{T} \, dU + \frac{P}{T}\, dV
$$
with $T$ fixed: the equation is integrable
since the function $V = V(P,T)$ and $U = U(P,T)$
are Lipschitz functions and hence their derivatives are bounded.
The resulting integral defines a function that is continuous even across
the Maxwell pressure hypersurface
$$
\{(U,S,T,P,V) \mid P = P_{\text{\rm mx}}(T)\}
$$
by the equal area law.

This enables us to express Gibbs free energy $\Upsilon$ on the equilibrium $R_{\CF;\CS}$
as a continuous function  $\Upsilon = f(T,P): = f_T(P)$ where $f$
is continuous but not differentiable across the curve
$\{(T,P)  \mid P = P_{\text{\rm mx}}(T) \}$:
The derivative $\frac{\del g}{\del P}$ jumps on $P = P_{\text{\rm mx}}(T)$
for $ 0 < T < T_c$. (See the highlighted part of Figure \ref{fig:graph4}.)
This finishes construction of a graph selector of
the equilibrium Legendrian submanifold $R_{\CF;\CS}$ over the
$PT$-plane, i.e., regarding $(P, T)$ as the preferred variables.

\begin{figure}[htbp]
\centering
\def\svgwidth{10pt}
\includegraphics[scale=0.5]{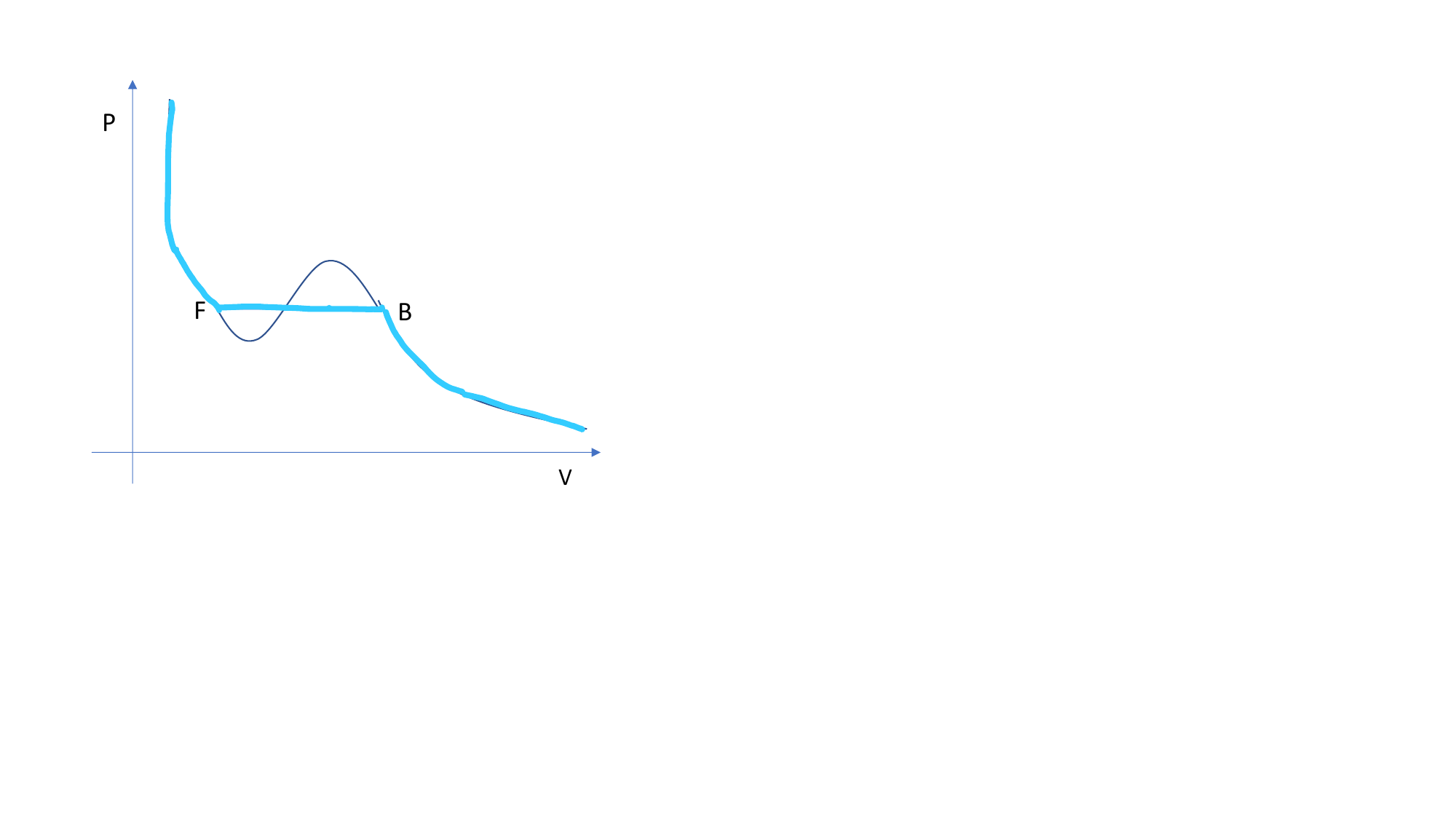}
\caption{Maxwell construction 2}
\label{fig:graph4}
\end{figure}

\end{proof}

The image of the one-jet map $j^1f_T: \R^2_{(T,P)} \to \R^5$ has
its  image contained in $R_{\CF;\CS}$ by construction with the
properties (See Figure \ref{fig:graph4}.):
\begin{enumerate}
\item  its cotangent projection  contains
jump-discontinuity across the set
$\{P = P_{\text{\rm mx}}\}$.
\item the function $f$ is defined on the open dense subset
$$
\pi_{(T,P)}\left(
R_{\CF;\CS} \setminus \{ P \neq P_{\text{\rm mx}}(T), \, 0 < T < T_c \}\right).
$$
\item the function, however, continuously extends across
the discontinuity locus
$$
\{(T,P) \mid P = P_{\text{\rm mx};T}, \, T \in \R_+\}.
$$
\item The aforementioned jump discontinuity  can be \emph{canonically} filled in by the Maxwell construction
on $\{ P \neq P_{\text{\rm mx}}(T), \, 0 < T < T_c \}$
using the equal area law.
\end{enumerate}

\begin{defn}[Maxwell adjustment] Denote by $R_{\CF;\CS}^{\text{\rm mx}}$
the above adjustment performed on the Legendrian submanifold
$R_{\CF;\CS}$ the Maxwell adjustment of the thermodynamic equilibrium. We denote by $R_{\CF;\CS}^{\text{\rm mx}}$ the
Maxwell adjustment.
\end{defn}
(See \cite[Section 3 \& 4]{oh:Lag-spectral} for detailed
construction of the general  Maxwell adjustment
and its explanation in general, and
Appendix \ref{sec:maxwell-adjustment} for a summary.)

\begin{rem}
\begin{enumerate}
\item The sudden jump of the volume $V = \frac{\del f_T}{\del P}$ at
the Maxwell pressure $P = P_{\text{\rm mx}}$ maintaining continuity of
the potential  $\Upsilon = H - TS = U + PV -TS$ across the pressure
has the following thermodynamic interpretation. The state whose pressure
is above that of the equilibrium $R_{\CF;\CS}$ describes a  liquid while
that below the Maxwell pressure describes vapour. Maxwell \cite{maxwell:construction}
himself described his construction of replacing the middle
sinusoidal part by the horizontal straight line in Figure \ref{fig:graph4}
by saying  \emph{``...... Since the temperature has
been constant throughout, no heat can have been transformed into work. Now the heat
transformed into work is represented by the excess of the area FDE over BCD. Hence the
condition which determines the maximum pressure of the vapour at given
temperature is that the line BF cuts off equal areas from the curve above and below.''}
\item
It is worthwhile to point out that the Maxwell construction is the
simplest scenario for which existence of such a continuous single-valued selector can be
directly established by simple integration. In the more complex system, establishing
existence of such a single-valued selector is a nontrivial problem in general,
as illustrated by \cite{oh:jdg}: The proof of existence of such a single-valued selector
involves either the Floer theory or the stable Morse theory in
symplectic geometry.
\end{enumerate}
\end{rem}

\section{Discussion on future direction: dynamics}
\label{sec:discussion}

We would like to mention that we have not touched upon the time
evolution or the dynamical point of view of nonequilibrium thermodynamics focus mostly on the geometric structure and
the kinematical aspect thereof.
For example, we have not touched upon anything about Boltzmann
analysis of involving the evolution equations
leading to equilibrium thermodynamics as $t \to \infty$.

The most common form of the Boltzmann equation
has the form
$$
\frac{\del f}{\del t} + v \cdot \nabla_x f = Q(f,f)
$$
in the position-velocity representation by $(x,v) \in \R^K \times \R^K$,
where $Q$ is the quadratic Boltzmann collision operator,
$$
Q(f,f) = \int_{\R^K} \int_{S^{K-1}}(f'f'_* - ff_*)B(v-v_*,\sigma)\, d\sigma
dv_*
$$
where $S^{K-1} \subset \R^K$ is the $K-1$ dimensional unit
sphere and $d\sigma$ is its standard volume form.
(see \cite[Equation (3)]{villani:kinetic} for meanings of the
unexplained notations $f'$, $f_*$, $f'_*$ and the Boltzmann's
collision kernel $B = B(v-v_*,\sigma)$. (In the notation of the current paper,
$v$ is $P$ and $x$ is $Q$.)

There are recent articles by Esen-Grmela-Pavelka
\cite{esen-grmela-pavelka-I,esen-grmela-pavelka-II} which investigate  the role of geometry in statistical mechanics and thermodynamics
in the themes somewhat similar to those studied in the present paper.
In addition, they unravel some interesting contact geometric aspect
and its role in the dynamical perspective. In particular, they put the full
kinetic equation, \emph{General Equation for Non-Equilibrium
Reversible-Irreversible Coupling} (GENERIC),  involving both Hamiltonian structure and gradient
structure into the framework of \emph{evolution Hamiltonian
dynamics}, i.e., dynamics generated by \emph{evolution
Hamiltonian vector field} \cite{SdLVdD}. In the course of their
investigation, they also consider
some mesoscopic dynamical theory, a bridge between
the statistical mechanics and the thermodynamics,  as some
sort of reduction, which they call the \emph{hierarchy reduction}.
Esen-Grmela-Pavelka write down a very
general form of the kinetic equation in \cite[Section 3.3]{esen-grmela-pavelka-II} which they acronym GENERIC:
Write $\rho = f\, \nu_0$ with  $f = \frac{\del \rho}{\del \nu_0}$. Then GENERIC has the following elegant form
$$
\dot f = \{f, E\} + \nabla_\beta \Xi|_{\beta = \delta_f S}
$$
$E$ is the energy of the system, $S$ is the Boltzmann
entropy and $\Xi = \Xi(f,\beta)$ is a general dissipation potential
defined on $T^*\CP(\Gamma)$, which satisfies
a list of properties \cite[Equation (9)]{esen-grmela-pavelka-II}
so that its evolution leads to
a thermodynamic equilibrium as $t \to \infty$.

It seems that our two-step process of reduction, following our
two mottos we take, are related to the Hamiltonian dynamics and the
gradient dynamics of their GENERIC framework respectively.
It remains to see if the study of the present
article would give some new light in the study of  the kinetic equation,
e.g., of Boltzmann's equation, and clarifying
the relationship between our study and that of
\cite{esen-grmela-pavelka-I,esen-grmela-pavelka-II} seems to
the first step towards that direction.

\appendix

\section{Information entropy}

The information entropy is firstly introduced by C. E. Shannon, under the attempt to research the communication in terms of mathematics \cite{shannon:entropy}.
Consider the probability measure space $(X,\CB,\mu)$.
We want to construct the function $I$ that
 measures the uncertainty of the information saying
$$
\text{``$x$ is in some `event' $A\in \CB$"}{}
$$
for any measurable subset $A\in \CB$.
The conditions for the information measure $I:\CB\rightarrow \mathbb{R}$ should be followings :
\begin{enumerate}
\item $I(A)$ is a function of $\mu(A)$
\item $I(A)\geq 0$
\item The smaller $\mu(A)$ is, the larger $I(A)$ is.
\item If $A,B$ are independent, then $I(A\cap B)= I(A)+I(B)$
\end{enumerate}
The condition (3) means that the uncertainty of the fact ``$x\in A$"increases as the size of $A$ decreases.
Intuitively, it reflects conviction that the statement ``I am an athlete" gives less information than
``I am an NBA player".

The condition (4) is from the ansatz that the two independent statements
connected with 'and'
give the same information obtained by considering the each statement individually.
Roughly speaking, consider the statement ``I'm an NBA player and he is an NFL player"
This statement just gives no more or less information from the two fact ``I'm an NBA player" and ``He is an NFL player".

\begin{lem}
The increasing continuous function $f:\mathbb{R}_+\rightarrow \mathbb{R}$ satisfying
$$
f(x)+f(y)=f(xy)
$$
should be $f(x)=-\log{x}$.
\end{lem}

By this lemma, the unique choice of $I$ satisfying condition (1)--(4) should be $I(A)=-\log{\mu(A)}$. Therefore, we define
the information measure function.

\begin{defn}
Let $(X,\CB,\mu)$ be a probability space. and $\CA$ be a subalgebra of $\CB$. Then, the \textbf{information measure} $I_\mu$ is defined by
$$
I_\mu(A)= -\log{\mu(A)}
$$
and the \textbf{expected information measure} $\widehat I (\mu)$ is defined by
$$
\widehat I_\mu(A)=\mu(A) I_\mu(A) = -\mu(A)\log{\mu(A)}
$$
for all $A\in \CA$.
\end{defn}

Suppose that $X$ is a finite set of events. Then, it is easy to define the expected information of probability distribution $\mu$, i.e
$$
S(\mu)=-\sum_{q^i \in X} \mu(q^i) \log \mu(q^i).
$$
However, if $X$ is infinite, the summation should become an integral, but the problem is that we do not have any measure on $X$ except the probability measure $\mu$.
Also, the logarithm term cannot be defined since $\mu$ does not
 give any pointwise value without comparing with other reference measure.
Therefore, we should consider a reference measure $\nu_0$
 in order to define the entropy of a probability  \cite{Jaynes:InfoStat}.
(This is the Liouville measure in the present article.)

\begin{defn}
Let $(X,\CB,\nu_0)$ be a measure space and $\rho$ be a probability measure which is absolutely continuous with respect to $\nu_0$. Then, the \textbf{relative entropy of $\rho$} is defined by
$$
\CS(\rho)= - \int_X \rho \log{\frac{d\rho}{d\nu_0}}
$$
where $\frac{d\rho}{d\nu_0}$ is a Radon-Nikodym derivative.
\end{defn}

\section{Relative information entropy and Legendrian generating function}
\label{sec:generating}

In this Appendix, we recall the definition of \emph{Legendrian generating function} and
explain what we mean by the statement
\emph{``Relative information entropy in $\SPS$ is a generating function of
thermodynamic equilibrium"}.  This framework is
motivated by  the following Weinstein's framework he observed on the
role of the classical action functional in the representation of Hamiltonian
deformations of the zero section of the cotangent bundle $T^*M$. To put
our framework in some perspective, we recall Weinstein's observation.

The standard generating function of a Lagrangian submanifold in the cotangent bundle
$T^*B$ is given by a vector bundle $\pi: E \to B$ equipped with a function $S: E\to \R$.
We will not repeat its definition in the symplectic context but refer readers to the
contact case below, except that we make the following remark.

Laudenbach-Sikorav \cite{laud-sikorav} and Sikorav \cite{sikorav}
proved that such a Lagrangian submanifold admits a generating function on a finite dimensional vector bundle
$\pi: E \to B$ by considering an approximation of the classical action functional
on the set of piecewise smooth paths which
approximates Hamiltonian trajectories issued at the zero section $o_{T^*B}$.
The procedure replaces
\emph{geodesics} in Bott's geodesic approximations of smooth loops \cite{bott:stable}
by the \emph{Hamiltonian paths}.

Then Weinstein observed that instead of considering this approximation,
one may directly put it into the following \emph{infinite dimensional} framework on
the general fibration, not just on the vector bundle.

\begin{lem}[Weinstein's lemma, \cite{alan:observation}]\label{lem:weinstein}
Let $H= H(t,x)$ be a time-dependent Hamiltonian on
 the cotangent bundle $T^*B$. Then the classical action functional
$$
\CA_H(\gamma) = \int \gamma^*\theta - H(t,\gamma(t))\, dt
$$
defined on the path space
$$
\CL_0(T^*B,0_{T^*B}) = \{ \gamma: [0,1] \to T^*B \mid \gamma(0) \in 0_{T^*B}\}
$$
is a `generating function' of the time-one image $\phi_H^1(0_{T^*B})$ of the zero section
under the Hamiltonian flow of $H$.
\end{lem}

More precisely, consider the diagram
\begin{equation}\label{eq:weinstein-diagram}
\xymatrix{\CL_0(T^*B,0_{T^*B}) \ar[d]^{\pi \circ \ev_1} \ar[r]^(.7){\CA_H} & \R \\
B &
}
\end{equation}
where $\ev_1: \CL_0(T^*B;0_{T^*B}) \to T^*B$ is the evaluation map
$$
\ev_1(c) : = c(1), \quad \text{\rm for } \, c \in \CL_0(T^*B;0_{T^*B}).
$$
Then the content and its proof of the above Weinstein's lemma
go as follows:
\begin{itemize}
\item First solve the fiberwise (or vertical) critical point equation. This provides
the \emph{equation of motion}, Hamilton's equation $\dot x = X_H(t,x)$.
\item Then push forward the image of the differential $d\CA_H$ by
the map $\pi \circ \ev_1$ to the cotangent bundle $T^*B$. \emph{This provides
the final location of the particle.}
\item The resulting push-forward image is precisely $\phi_H^1(0_{T^*B})$.
\emph{This provides the final momentum of the particle.}
\end{itemize}

Now we recall the formal definition of Legendrian generating functions
in the context of contact geometry.
We refer to \cite{sandon:generating} for a good introduction thereto in relation to
the purpose of the present section.

Let $B$ be a  manifold and $J^1B = T^*B \times \R$ be the
one-jet bundle. We equip $J^1B$ with the contact form
$$
\lambda = dz - pdq =: dz - \pi^*\theta
$$
where $(q,p,z) =: y$ be the canonical coordinates of $J^1B$.
We also denote $x = (q,p) \in T^*B$.

We denote by $\pi:E \to B$ a finite dimensional vector bundle.

\begin{defn} We say a Legendrian submanifold $R \subset J^1B$ is generated by
a function $S: E \to \R$ if we can express
$$
R = \{(q,p,z) \mid p = D^hS, \, z = S(q,e), \, d^vS(q,e) = 0\}.
$$
When this holds, we say $S$ is a generating function of the Legendrian submanifold $R$.
\end{defn}

Now we apply the above discussions to
our diagram \eqref{eq:generating-diagram}
\begin{equation}\label{eq:generating-diagram3}
\xymatrix{ {\CM^\CF}  \ar[d]_{\pi_\CF} \ar[r]^>>>>>{\CS^{\text{\rm red}}_\CF} & \R \\
\mathfrak g_{\CO_\CF}^* &
}
\end{equation}
This diagram is in the same infinite dimensional spirit as Weinstein's diagram \eqref{eq:weinstein-diagram}.
Here we break down the procedure of determining the
thermodynamic equilibrium and the associated probability distribution:

\begin{enumerate}[(1)]
\item (Finite dimensional reduction) The fiberwise (or vertical) critical point equation provides
the \emph{equation of motion}, which is the \emph{constrained} extremization problem
$$
d^v(\CS_\CF^{\text{\rm red}})([(\rho,\beta)]) = 0.
$$
\item (Generating Lagrangian shadow)
Then we push forward the image of the differential $d\CS_\CF^{\text{\rm red}}$ by
the map $\pi_\CF$ to the cotangent bundle $T^*\mathfrak g_\CF^*$.
\item (Generating thermodynamic equilibrium)
The canonical lifting to $J^1 \mathfrak g_\CF^*$ provides a
thermodynamic equilibrium state which satisfies the first law of
thermodynamics.
\item (Determination of equilibrium probability distribution)
The \emph{full} extremization problem over all observation data which
is equivalent to the \emph{full} critical point equation of $\CS_\CF^{\text{\rm red}}$
$d(\CS_\CF^{\text{\rm red}})(\rho) = 0$ is now reduced to the
finite dimensional problem of solving
\begin{equation}\label{eq:finite-dim-reduction}
D^h(\CS_\CF^{\text{\rm red}})([(\rho,\beta)]) = 0
\end{equation}
on $R_{\CF;\CS}$.
\end{enumerate}

Some thermodynamic interpretation of each step of the above is now in order.

What we have shown in Subsection \ref{sec:reduced-entropy}
is the statement that the $\CF$ iso-data mesoscopic KTPS $\CE_{\CF\CS}/\CG$
is the collection of states after performing the finite dimensional
reduction  of solving $d^v(\CS_\CF^{\text{\rm red}}) = 0$:
This corresponds to the above steps (1) and (2).

Step (3) above corresponds to the step of mapping $\CE_{\CF\CS}/\CG_\CF$ into
the finite dimensional TPS $J^1 \mathfrak g_\CF^*$ in the way that
the points in the image $R_{\CF;\CS}$ of the map $[\text{\rm pr}]$
satisfy the first thermodynamic law $dw - \sum_{i=1} p_i dq_i = 0$.
By the dimensional reason $R_{\CF;\CS}$ becomes a Legendrian submanifold.
This is the step (3) above.

When $R_{\CF;\CS}$ is holonomic,
in which case $R_{\CF;\CS} = \Image j^1f$ for some \emph{smooth} function
$f: \mathfrak g_\CF^* \to \R$.
the finite-dimensional reduction \eqref{eq:finite-dim-reduction} becomes a
finding a critical point of the relevant thermodynamic potential $f$. Otherwise
one has to solve the extremization problems
of a \emph{Lipschiz continuous}, not necessarily differentiable function, after applying
the Maxwell construction and selecting a single-valued branch of the
non-holonomic equilibrium state $R_{\CF;\CS}$. This process has been explained
in more detail with the van der Waals model in Section \ref{sec:maxwell}.

\section{General construction of Maxwell adjustment}
\label{sec:maxwell-adjustment}

In this section, we recall the details of construction of Maxwell adjustment
from \cite{oh:Lag-spectral} (see also \cite{arnaud} and \cite{amorim-oh-santos})
for a general compact exact Lagrangian submanifold in the cotangent bundle $T^*N$.
It is proved in \cite{oh:jdg}, \cite{amorim-oh-santos} that such a Lagrangian submanifold carries a
canonical graph selector $f_L$ called the \emph{basic phase function},
and that the structure of its differential $\sigma_F$, called
the \emph{Lagrangian selector}, is analyzed in \cite{oh:Lag-spectral} . Such an existence of graph selector has been constructed in
\cite{oh-yso:J1B} for a general compact Legendrian submanifold contact isotopic to the zero section of
one-jet bundle \cite{oh-yso:J1B}.

This being said we introduce the following general definition motivated by this.

\begin{defn} Assume $N$ is a closed manifold. Let $R \subset J^1N$ be any compact
smooth Legendrian submanifold projecting
surjectively to $N$. We call a (densely defined) single-valued function $f:N \to \R$  a \emph{graph selector}
and its differential $\sigma: = df$ of $T^*N$ with values
lying in $L = \pi_{T^*N}(R)$ its associated \emph{Lagrangian selector} of $R$.
\end{defn}
$R$ carries two projections, the projection of $R$ to $N \times \R$ called
the \emph{front projection} $\pi_{\text{\rm front}}$ and projection of $R$ to $T^*N$,
called the Lagrangian projection. We denote by $L$ the image of \emph{Lagrangian projection} or \emph{cotangent projection}.

We make the following hypothesis for the further discussion below, which we know holds
for the Legendrian submanifold contact isotopic to the zero section of $J^1N$ from
\cite{oh-yso:J1B}.

\begin{hypo}\label{hypo:selector-exists} We assume that
the projection of $R$ to $N \times \R$ is surjective and that
$R$ admits a graph selector $f: N \to \R$ and denote by
$\sigma: = df$ the associated Lagrangian selector.
\end{hypo}
The differential $\sigma$ is not a continuous map but is a well-defined current.
Under this hypothesis, the construction given in \cite[Section 3 \& 4]{oh:Lag-spectral}
immediately generalizes to such Legendrian submanifolds as follows.

For a given graph selector $f: N \to \R$ given by Hypothesis \ref{hypo:selector-exists},
we denote by $S(\sigma)$ the singular set of
the current $\sigma$. Then we consider the open subset of $R$
$$
\Sigma_R: = \{(q,z) \in J^1N \mid q \in N \setminus S(\sigma_R),
\, p = \sigma_R(q), \,(q,p,z) \in R \} \subset R.
$$
The projection $\pi_R: R \to N$ restricts to a one-one
correspondence on $\Sigma_R \subset R$, and the smooth function
$$
f|_{N_0}:  N_0 \to \R
$$
continuously extends to $\overline{N}_0 = N$ where $N_0 = N \setminus S(\sigma)$.
 By construction, we have the bound
\begin{equation}\label{eq:dfL}
 |df_R(q)| \leq \max_{x \in R} |p(x)|
\end{equation}
for any $q \in N_0$, where $x = (q(x),p(x),z(x))$ and the norm
$|p(x)|$ is measured by any given Riemannian metric on $N$.

The general structure theorem of the wave front (see
\cite{eliash:front}) proves that the section
$\sigma$ is a differentiable map on a set of full measure for a
generic choice of $R$ which is, however, \emph{not necessarily
continuous}.

We note that the singular locus $S(\sigma) \subset N$ is a subset of
the \emph{bifurcation diagram} of the Legendrian submanifold $R$:
The bifurcation diagram is the union of the caustic and the \emph{Maxwell set}
where the latter is the set of points of which the different branches of the
generating function merge. (See \cite[Section 4]{givental:singular} for
the definition of bifurcation diagram of Lagrangian submanifold $L \subset T^*N$
which can be generalized to the Legendrian submanifold of the type we consider here
in general.)

For a generic $R$, $S(\sigma_R)$ is stratified into a finite union of smooth submanifolds
$$
\bigcup_{k=1}^{n} S_k(\sigma_R), \quad S_k(\sigma_R)= \Sing_k(\sigma_R), \quad n = \dim N
$$
(see \cite{arnold:normalform,eliash:front,givental:singular} e.g., for such a result)
so that its conormal variety $\nu^*S(\sigma_R)$ can be defined as a finite union of
conormals of the corresponding strata.
Each stratum $S_k(\sigma_R)$ has codimension $k$ in $N$.
The stratum for some $k$ could be empty.
In $\dim N = 2$, there are two strata to consider,
one $S_1(\sigma_R)$ and the other $S_2(\sigma_R)$.

For $k=1$, $S_1(\sigma) \subset N$ is a hypersurface and each given point $q \in S_1(\sigma_R)$ has a neighborhood
$U(q) \subset N$ such that $U(q) \setminus S_1(\sigma_R)$
has two components.
We also note that $\sigma_R$ carries a natural orientation induced from $N$
by projection when $N$ is orientable.
When $N$ is oriented, $S_1(\sigma_R)$ is also
orientable as a finite union of smooth hypersurface. We fix any orientation on $S_1(\sigma_R)$.

We denote by $U^\pm(q)$ the closure of
each component of $U(q) \setminus S_1(\sigma_R)$ in $U(q)$ respectively.
Here we denote by $U^+(q)$ the component whose boundary orientation
on $\del U^+(q)$ coincides with that of the given orientation on
$S_1(\sigma_R)$
and by $\del U^-(q)$ the other one.
Then each of $U^\pm(q)$ is an open-closed domain with the same boundary
$$
\del U^\pm(q) = U(q) \cap S_1(\sigma_R).
$$
Denote
\begin{equation}\label{eq:fpmatq}
\sigma^\pm_R(q) = \lim_{p_\pm \to q} \sigma_R(p_\pm)
\end{equation}
obtained by taking the limit on $U^\pm(q)$ respectively. The limits
are well-defined from the definition of $\sigma_R$ since $R$ is a smooth
manifold.

The following theorem is proved in \cite{oh:Lag-spectral}. See also
\cite{givental:singular}, \cite{zak-roberts} for a related statement.

\begin{prop}[Theorem 4.1 \cite{oh:Lag-spectral}]\label{prop:conormal} Let $q \in S_1(\sigma_R)$.
Then
$$
\sigma^-_R(q) - \sigma^+_R(q) \in T_q^*N,
$$
is contained in the conormal space $\nu_q^*[S_1(\sigma_R);N] \subset T_q^*N$.
\end{prop}

The boundary orientations
of the two components arising from that of $\Sigma_R$, which in turn
is induced from that of $N$ via the cotangent projection have opposite orientations. We call
the one whose projection to $S_1(\sigma_R)$ coinciding with the given orientation
the \emph{upper branch} and the one with the opposite one the \emph{lower branch}
and denote them by
$$
\del^+\Sigma_R, \, \del^-\Sigma_R
$$
respectively.

Now let $\ell_q$ be the line segment connecting the two vectors
$\sigma^\pm_R(q)$, i.e.,
\begin{equation}\label{eq:Lq}
\ell_q: u \in [0,1] \mapsto  \sigma^+_R(q) + u(\sigma^-_R(q) - \sigma^+_R(q)) \subset T_q^*N.
\end{equation}
This is an affine line in the vector space $T_q^*N$ that is parallel to the conormal space $\nu^*_{q}S_1(\sigma_R)$.
Therefore the union
\begin{equation}\label{eq:unionLq}
\Sigma_{R;[-+]}: = \bigcup_{q \in S_1(\sigma_R)} \ell_q
\end{equation}
is contained in the translated conormal space
\begin{equation}\label{eq:trans-conormal+}
\sigma_R^+ + \nu^*[S_1(\sigma_R);N]
\end{equation}
Here the bracket $[-+]$ stands for the line segment $\ell_q$, and
 $\nu^*[S_1(\sigma_R);N]$ is the conormal bundle of $S_1(\sigma_R)$ in $N$.
We would like to point out that since $\sigma_R^+(q) - \sigma_R^-(q) \in \nu^*[S_1(\sigma_R);N]$
we have the equality
$$
\sigma_R^+(q) + \nu_q^*[S_1(\sigma_R);N] = \sigma_R^-(q) + \nu_q^*[S_1(\sigma_R);N]
$$
for all $q \in S_1(\sigma_R)$.
Therefore we can simply write \eqref{eq:trans-conormal+} as
\begin{equation}\label{eq:trans-conormal}
\sigma_R + \nu^*[S_1(\sigma_R);N]
\end{equation}
unambiguously. The following definition was called the \emph{basic Lagrangian selector chain}
in \cite{oh:Lag-spectral} as a singular chain.

\begin{defn}
We denote by $L_\sigma^{\text{\rm mx}}$ the Lagrangian chain
\begin{equation}\label{eq:SigmaF}
\overline \Sigma_R \subset T^*N
\end{equation}
with the orientation given as above.
\end{defn}

The two components of $\del \overline \Sigma_R$ associated to each connected
component of $S_1(\sigma_R)$ are the graphs of $\sigma_R^\pm$ for the functions
$f_F^\pm$ near $S_1(\sigma_R)$. Again we regard $\sigma_R^\pm$ as singular chains.

Note that each connected component of $S_1(\sigma_R)$ gives rise to
two components of $\del \Sigma_{R;[-+]} \cap \Sigma_R$.
We can bridge the `cliff' between the two branches
of $\del \Sigma_R$ over each connected component of $S_1(\sigma_R)$.

\begin{defn}[Cliff wall chain] We
define a `cliff wall' chain $\sigma_{R;[-+]}$ whose support is given by the union
$$
\Sigma_{R;[-+]} = \bigcup_{q \in S_1(\sigma_R)} \ell_q
$$
Then we define the chain
$\Sigma_{R;[-+]}$ similarly as we define $\Sigma_R$ by taking its closure in $T^*N$.
\end{defn}

By definition, its tangent space at $x = (q,u)$ has natural identification with
$$
T_x \Sigma_{F;[-+]} \cong \nu^*_q S_1(\sigma_R) \oplus T_q S_1(\sigma_R).
$$
Thanks to Proposition \ref{prop:conormal}, it carries a natural direct sum orientation
$$
o_{\Sigma_{R;[-+]}}(q) = \{df^-_R(q) - df^+_R(q)\} \oplus o_{S_1(\sigma_R)}(q).
$$
Therefore $\Sigma_{R;[-+]}$ carries a natural orientation and defines a singular chain.
(Under the natural identification of $T_qN$ with $T_q^*N$ by the dual pairing,
which induces an identification
$$
\nu^*_q S_1(\sigma_R) \oplus T_q S_1(\sigma_R) \cong \nu_q S_1(\sigma_R) \oplus T_q S_1(\sigma_R)
$$
as an oriented vector space.) Then we have the relation
\begin{equation}\label{eq:orient}
\del \Sigma_R = - \del \Sigma_{R;[-+]}
\end{equation}
along the intersection $\del \Sigma_R \cap \del \Sigma_{R;[-+]}$. The above discussion
leads us to the following

\begin{prop}[Maxwell adjustment]
Suppose a compact Legendrian submanifold $R \subset J^1N$ satisfies Hypothesis \ref{hypo:selector-exists},
and let $f: N \to \R$ be a graph selector.
We consider the union
$$
L_f^{\text{\rm mx}}: = \Sigma_R \cup \Sigma_{R;[-+]}.
$$
Then the set
$$
R_f^{\text{\rm mx}}: = \{(q,p,z) \in J^1N \mid (q,p) \in L_f^{\text{\rm mx}}, \, z = f(q)\}.
$$
defines a Legendrian cycle. We call $R_f^{\text{\rm mx}}$ the \emph{Maxwell adjustment} of
the Legendrian submanifold $R$ associated to $f$.
\end{prop}

This finishes our
discussion on the Maxwell construction for the general Legendrian submanifold in
the one-jet bundle $J^1N$ of general closed manifold $N$.
\medskip

\noindent{\bf Acknowledgements:} The second named author thanks Sang-Jin Sin for his
interest in this work and illuminating discussions on the nature of volume variable
$V$ which greatly helps in finding our formulation of the volume, a nonlocal thermodynamic
observable, from a similar kind of viewpoint applied to the local
observables. We thank the referees for providing many useful
comments and suggestions which help us to much improve
the presentation of the present article. 
This work is supported by the IBS project \# IBS-R003-D1.

\medskip

\noindent{\bf Competing interest:} There is no competing interest.

\bibliographystyle{plain}
\def\cprime{$'$}

\end{document}